\documentclass[a4paper, USenglish,cleveref, thm-restate]{lipics-v2021}

\pdfoutput=1 %
\hideLIPIcs  %

\title{Optimal antimatroid sorting}

\author{Benjamin Aram Berendsohn}{Max Planck Institute for Informatics, Saarbrücken, Germany}{benjamin.berendsohn@fu-berlin.de}{https://orcid.org/0000-0002-3430-5262}{}

\authorrunning{B.\,A. Berendsohn}

\Copyright{Benjamin Aram Berendsohn}



\category{} %

\relatedversion{} %

\nolinenumbers %

\EventEditors{Anne Benoit, Haim Kaplan, Sebastian Wild, and Grzegorz Herman}
\EventNoEds{4}
\EventLongTitle{33rd Annual European Symposium on Algorithms (ESA 2025)}
\EventShortTitle{ESA 2025}
\EventAcronym{ESA}
\EventYear{2025}
\EventDate{September 15--17, 2025}
\EventLocation{Warsaw, Poland}
\EventLogo{}
\SeriesVolume{351}
\ArticleNo{103}

\usepackage{algorithm}
\usepackage[noEnd,indLines]{algpseudocodex}
\usepackage{cite}
\usepackage{framed}
\usepackage{ifthen}
\usepackage{mathtools}
\usepackage{tikz}

\usepackage{cleveref}

\theoremstyle{definition}
\newcommand{\R}{\mathbb{R}}
\newcommand{\N}{\mathbb{N}}
\newcommand{\fO}{\mathcal{O}}
\newcommand{\gminus}{-} %

\newcommand{\tctotal}{t^{\mathrm{total}}}
\newcommand{\tq}{t^{\mathrm{Q}}} %
\newcommand{\afunc}[1]{\Call{#1}{}}

\algrenewcommand\textproc{\normalfont\texttt} %

\makeatletter
\algnewcommand\algorithmicforeach{\textbf{for each}}
\algdef{S}[FOR]{ForEach}[1]{%
	\algpx@startCodeCommand\algpx@startIndent\algorithmicforeach\ #1\ \algorithmicdo%
}
\pretocmd{\ForEach}{\algpx@endCodeCommand}{}{}
\makeatother

\newcommand{\cdsinit}{\afunc{init}}
\newcommand{\cdsstep}{\afunc{step}}
\DeclareMathOperator{\Av}{Av} %

\newcommand{\emptyword}{\varepsilon}
\newcommand{\support}[1]{\tilde{#1}}
\newcommand{\swords}[1]{#1_{\mathrm{s}}^*}
\newcommand{\perms}{\mathcal{P}} %
\newcommand{\Lang}{\mathcal{L}}

\newcommand{\topos}{\mathcal{T}} %
\DeclareMathOperator{\Cont}{Cont}

\newcommand{\AMVS}{A^{\mathrm{VS}}} %
\newcommand{\AMSVP}{A^{\mathrm{SVP}}} %

\begin{document}
	\maketitle
	\begin{abstract}
		The classical \emph{comparison-based sorting} problem asks us to find the underlying total order of a given set of elements, where we can only access the elements via comparisons. In this paper, we study a \emph{restricted} version, where, as a hint, a set $T$ of possible total orders is given, usually in some compressed form.

		Recently, an algorithm called \emph{topological heapsort} with optimal running time was found for the case where $T$ is the set of topological orderings of a given directed acyclic graph, or, equivalently, $T$~is the set of linear extensions of a given partial order [Haeupler et al.\ 2024]. We show that a simple generalization of topological heapsort is applicable to a much broader class of restricted sorting problems, where $T$ corresponds to a given \emph{antimatroid}.
		
		As a consequence, we obtain optimal algorithms for the following restricted sorting problems, where the allowed total orders are \dots
		\begin{itemize}
			\item \dots restricted by a given set of \emph{monotone precedence formulas};
			\item \dots the \emph{perfect elimination orders} of a given chordal graph; or
			\item \dots the possible \emph{vertex search orders} of a given connected rooted graph.
		\end{itemize}
	\end{abstract}
	
	\section{Introduction}\label{sec:intro}
	
	One of the fundamental problems in theoretical computer science is \emph{sorting}. It has a wide variety of applications in practice~\cite{Knuth1973} and is frequently used as an introduction to design and analysis of algorithms.

	Perhaps the most intensely studied variant is \emph{comparison-based sorting}, where one can only access the input elements via comparisons (otherwise, we assume the standard word RAM model). This excludes bucket sort, radix sort, and other integer-sorting techniques~\cite{FredmanWillard1993,HanThorup2002,Han2004}. The well-known \emph{information-theoretic bound (ITB)} states that comparison-based sorting requires $\log(n!) \approx n \log n$
	comparisons in the worst case, and many algorithms asymptotically matching this bound are known.
	
	In this paper, we study a restricted variant of comparison-based sorting.
	In addition to the $n$ elements to be sorted, we are give a set $T$ of possible total orders, with the guarantee that the actual order of the input is contained in $T$. Call this the \emph{$T$-sorting problem}. Note that $T$ may be very large compared to $n$, so it makes sense to consider variants of the problem where $T$ is given in some compressed form. Regardless of that, the ITB for the $T$-sorting problem is $\log |T|$, so perhaps the first question to ask is:
	Can we solve the $T$-sorting problem with only $\fO( \log |T| )$ comparisons?
	
	In turns out that the answer is \emph{yes} if $|T| \ge 2^{\Omega(n)}$. Fredman~\cite{Fredman1976} showed how to solve the $T$-sorting problem with $\fO( n + \log |T| )$ comparisons. This algorithm assumes that $T$ is given explicitly and is thus entirely impractical.
	Fredman also showed that there exist sets $T$ with $|T| \in \fO(n)$ where the $T$-sorting problem requires $\Omega(n)$ comparisons, which means the ITB is not always tight.

	Most later work on this problem has been focused on the special case where $T$ is the set of linear extensions of a partial order~$P$, given by its Hasse diagram. It is known that $\fO(\log |T|)$ comparisons are sufficient, that is, the ITB is tight (even when $\log T \le n$). After decades of research~\cite{Kislitsyn1968,Fredman1976,Linial1984,KahnKim1992,BrightwellFelsnerEtAl1995,CardinalFioriniEtAl2010,vanderHoogRutschmann2024a}, eventually, two algorithms were found that are simultaneously comparison-optimal and running-time optimal~\cite{HaeuplerHladikEtAl2025a,vanderHoogRotenbergEtAl2025}. Both algorithms actually solve the slightly more general \emph{DAG sorting} problem: Given is a directed acyclic graph $G$, and $T$ is taken as the set of topological orderings of $G$. Note that each DAG $G$ defines a partial order $P$ (though $G$ may contain some ``unnecessary'' edges), and $T$ is the set of linear extensions of $P$.
	
	Most important for our purposes is the \emph{topological heapsort} algorithm of Haeupler, Hladík, Iacono, Rozhon, Tarjan, and Tětek~\cite{HaeuplerHladikEtAl2025a}. This was the first DAG-sorting algorithm with optimal running time $\fO(|V(G)| + |E(G)| + \log |T| )$, though it makes $\Theta(|V(G)| + \log |T| )$ comparisons, which is non-optimal if $|T|$ is small. Achieving comparison optimality requires further modifications, which we discuss later.
	
	We briefly describe topological heapsort.
	Let $G$ be the given DAG. At the beginning, identify the set $S$ of sources (vertices of in-degree 0) of $G$, and insert them into a priority queue $Q$. In each following step, we first find and remove the minimum element $x$ from $Q$. Then, we remove $x$ from $G$. This may create several new sources, which we all insert into $Q$. Repeat until $Q$ is empty.
	It is easy to see that this algorithm is \emph{correct} if the underlying total order of the vertices is indeed a topological ordering of $G$.
	Haeupler et al.\ achieve optimal running time %
	by using a special priority queue with the so-called \emph{working-set property}.
	
	Let us now try to generalize topological heapsort to solve the general $T$-sorting problem. Suppose that instead of the graph $G$, we are given a black-box data structure $D$ that reports all elements that are the minimum of at least one ordering $\tau \in T$. After ``removing'' the first $k$ elements $x_1, x_2, \dots, x_k$, let $D$ report all elements $y$ such that $x_1 x_2 \dots x_k y$ is a prefix of some ordering $\tau \in T$ (the \emph{available} elements). Again, it is easy to see that topological heapsort with this data structure is \emph{correct}. But is it optimal?
	
	The main contribution of this paper (\cref{sec:main}) is to show that topological heapsort has optimal running time if $T$ corresponds to a class of structures called \emph{antimatroids}, which are a generalization of partial orders, when we disregard the running time for the data structure~$D$. %
	We additionally show that, for several possible representations of antimatroids, we can implement $D$ such that its total running time is linear in the size of the input, so the running time stays optimal (\cref{sec:app}). Then, we show how to additionally achieve comparison optimality (\cref{sec:comparisons}). At the end, we discuss several generalizations of antimatroids (\cref{sec:generalizations}) where topological heapsort is \emph{not} optimal.
	
	In the remainder of this section, we give an informal definition of antimatroids, and discuss some related work.
	
	\subparagraph{Antimatroids.}
	Fix a set $T$ of total orders, and consider the black-box data structure $D$ mentioned above. Recall that, for each sequence $x_1, x_2, \dots, x_k$ of $k \ge 0$ elements, $D$ tells us the set of \emph{available} elements $y$ such that $x_1 x_2 \dots x_k y$ is a prefix of some $\tau \in T$. It is clear that any set $T$ can be described in this way.
	Intuitively, $T$ is an \emph{antimatroid} if the following conditions apply. First, once an element becomes available, it stays available until it is chosen (availability is \emph{monotone}); and second, availability of an element can only depend on the \emph{set} of elements chosen so far (but not their order).
	
	The first condition can be nicely related to topological heapsort: It means that our priority queue $Q$ at no point contains unavailable elements.
	The second condition may seem more surprising, but it turns out to be also necessary for optimality of topological heapsort and tightness of the ITB (see \cref{sec:generalizations}).
	Antimatroids were first studied by Dilworth~\cite{Dilworth1940} in lattice theory. Later, they were found to be special cases of so-called \emph{greedoids}~\cite{KorteLovasz1981,KorteLovasz1984a}, which are structures admitting a certain greedy optimization procedure. A common alternative definition of antimatroids treats them as \emph{set systems}; this definition can be related to \emph{matroids} (hence the name), which are also special cases of greedoids. We refer to the surveys of Björner and Ziegler~\cite{BjoernerZiegler1992} and Korte, Schrader, and Lovász~\cite{KorteSchraderEtAl1991} for more.

	\subparagraph{Related work.}

	Optimal algorithms are known~\cite{HwangLin1972,BrownTarjan1979,Linial1984} for \emph{merging} two sorted lists; this corresponds to a partial order consisting of exactly two disjoint chains.
	Merging can be generalized to \emph{multiway merging}, where the input is split into any number of sorted lists; here, the ITB is known to be related to the \emph{Shannon entropy} of the list lengths. Multiway merging is important for pratical sorting algortihms that exploit \emph{runs} in a given unsorted sequence (see, e.g., Gelling, Nebel, Smith, and Wild~\cite{GellingNebelEtAl2023a}).
	
	A different special case of the $T$-sorting problem concerns the case where the set $T$ is the permutations that \emph{avoid} a certain pattern $\pi$.\footnote{A permutation $\tau$ of a set $X$ \emph{avoids} a permutation $\pi$ of $[k]$ if $\tau$ contains no subsequence that is order-isomorphic to $\pi$, according to some fixed total order on $X$. This restricted sorting problem is usually formulated in a different, but equivalent way where $X$ is given as a sequence that is known to avoid a pattern.}
	It is known that $T \le c_\pi^n$ for some constant $c_\pi$ depending on $\pi$, so the ITB is linear in this case. Opler~\cite{Opler2024} gave an optimal linear-time algorithm for the pattern-avoiding sorting problem, which works even when the pattern $\pi$ is unknown. We refer to his paper for more information on that problem.
	
	The \emph{$X+Y$-sorting problem} is a restricted sorting problem where, interestingly, the ITB is \emph{not} tight. Here, two sets $X, Y \subset \R$ of size $k$ each are given, and we need to sort the set $\{x+y \mid x \in X, y \in Y\}$ of size $n = k^2$. If only comparisons of the form $x + y < x' + y'$ with $x,x' \in X$, $y,y' \in Y$ are allowed, then this is a restricted sorting problem by our definition. It is known that the number of possible orderings is only $k^{\Theta(k)}$~\cite{Fredman1976}, which makes the ITB $\Omega(k \log k) = \Omega( \sqrt{n} \log n)$; however, Fredman~\cite{Fredman1976} showed that $\Omega(k^2)$ comparisons are needed. It is currently unknown if an algorithm with optimal running time $\Theta(k^2)$ exists.
		
	Finally, a question closely related to restricted sorting is finding \emph{balanced pairs}, which are comparisons that split the set of possible total orders into two parts of approximately equal size. There is a body of work on balanced pairs in partial orders~\cite{Kislitsyn1968,Fredman1976,Linial1984,BrightwellFelsnerEtAl1995,Brightwell1999}, and the concept has been generalized to antimatroids~\cite{Eppstein2014}.
	
	\section{Preliminaries}\label{sec:pre}

	All the following definitions are taken from Björner and Ziegler~\cite{BjoernerZiegler1992} with only slight changes.
	
	\subparagraph{Words and languages.}
	Let an \emph{alphabet} $\Sigma$ be a set of \emph{letters}. A \emph{word} on $\Sigma$ is a finite sequence of letters. We write words without commas as $\alpha = x_1 x_2 \dots x_n$, where $x_i \in \Sigma$, and we write $\alpha \beta$ for the concatenation of two words $\alpha$ and $\beta$. We denote the length of a word $\alpha$ by $|\alpha|$. The \emph{support} $\support{\alpha}$ of a word $\alpha$ is the set of letters contained in it, and the empty word is denoted by $\emptyword$.
	
	A word $\alpha$ is \emph{simple} if each letter occurs at most once in it (i.e., $|\support{\alpha}| = |\alpha|$), and $\alpha$ is a \emph{permutation} of $\Sigma$ if each letter occurs precisely once in it. The set of all simple words on $\Sigma$ is denoted by $\swords{\Sigma}$, and a \emph{simple language} is a subset $L \subseteq \swords{\Sigma}$. We denote the set of permutations contained in a simple language by $\perms(L)$, and we call $L$ \emph{full} if $|\perms(L)| > 0$.

	\subparagraph{Partial orders.}
	A partial order on a set $\Sigma$ is denoted with a capital letter, like~$P$, and we write $\prec_P$ for the actual partial order relation. %
	Total orders are treated as permutations of the underlying set $\Sigma$. Whenever necessary, we write $\prec_\pi$ for the total order relation corresponding to the permutation $\pi$, i.e., we have $x \prec_\pi y$ if $x$ precedes $y$ in $\pi$. An \emph{oracle} for a total order $\pi$ is a black-box function that answers queries of the form $x \prec_\pi y$ in constant time.

	\subparagraph{Antimatroids.}

	An \emph{antimatroid} (sometimes called \emph{antimatroid language}) is a simple language $A \subseteq \swords{\Sigma}$ that satisfies the following two axioms:
	
	\begin{enumerate}[(i)]
		\item For all $\alpha \in \swords{\Sigma}$ and $x \in \Sigma \setminus \support{\alpha}$, if $\alpha x \in A$, then $\alpha \in A$. ($A$ is closed under taking prefixes.)\label{prop:am:acc}
		\item For each $\alpha, \beta \in A$ with $\support{\alpha} \not\subseteq \support{\beta}$, there is some $x \in \support{\alpha} \setminus \support{\beta}$ such that $\beta x \in A$.\label{prop:am:exchange}
	\end{enumerate}
	
	Axiom (\ref{prop:am:acc}) is called \emph{accessibility} and implies that every word in $L$ can be \emph{built} by starting with $\emptyword$ and successively appending letters, without ever leaving $L$.
	Axiom (\ref{prop:am:exchange}) ensures the two ``availability properties'' of antimatroids mentioned in the introduction.

	In this paper, we only consider \emph{full} antimatroids (recall that this means $\perms(A) \neq \emptyset$).\footnote{Some authors define antimatroids to be always full~\cite{KorteSchraderEtAl1991}.} A~full antimatroid $A$ is completely determined by $\perms(A)$; in fact, $A$ is exactly the set of prefixes of $\perms(A)$.
	However, not every set of permutations can be described by an antimatroid; for example, any antimatroid on $\Sigma = \{x,y,z\}$ that contains $xyz$ and $zyx$ also contains $xzy$. %

	Throughout this paper, we use a different formalism, which is closer to the informal definition of antimatroids given in the introduction. For an alphabet $\Sigma$, a \emph{precedence function} for a letter $x \in \Sigma$ is a function $p_x \colon 2^{\Sigma \setminus \{x\}} \rightarrow \{0,1\}$ that is \emph{monotone}, i.e., we have $p_x(X) \le p_x(Y)$ if $X \subseteq Y \subseteq \Sigma \setminus \{x\}$. A \emph{monotone precedence system (MPS)} on $\Sigma$ is a collection $S = \{p_x\}_{x \in \Sigma}$ of precedence functions. The \emph{language} of $S$ is the set
	\[
		\Lang(S) = \{ x_1 x_2 \dots x_k \mid \forall i \in [k] : p_{x_i}(\{x_1, x_2, \dots, x_{i-1}\}) = 1 \}.
	\]
	
	We show in \cref{app:mps-char} that monotone precedence systems indeed characterize (full) antimatroids.
	In the following, we write $\perms(S) = \perms(\Lang(S))$ for convenience. %
	
	\begin{example}\label{ex:poset}
		Let $P$ be a partial order on a set $\Sigma$. Define  $S = \{p_x\}_{x\in \Sigma}$ with $p_x(X) = 1$ if and only if $\{ y \in \Sigma \mid y \prec_P x \} \subseteq X$. It is easy to see that $\perms(S)$ is exactly the set of linear extensions of $P$. Thus, antimatroids generalize partial orders.
	\end{example}
	
	\begin{example}
		Let $\Sigma = \{a,b,c\}$, and consider the MPS $S$ with $p_a(X) = p_b(X) = 1$ for all $X \subseteq \Sigma$, and $p_c(X) = 1$ iff $a \in X$ or $b \in X$. Then $\perms(S) = \{ abc, bac, acb, bca \}$. Since $\perms(S)$ contains both $acb$ and its reverse $bca$, but not all possible permutations, it is \emph{not} the set of linear extensions of any partial order. Thus, antimatroids \emph{strictly} generalize partial orders.
	\end{example}
	
	\subparagraph{Priority queues.} The central data structure of topological heapsort is the \emph{working-set priority queue}. This is a priority queue (like the well-known \emph{binary heap}~\cite{Williams1964}) with the so-called \emph{working-set property}. Informally, this property ensures that extracting an element is fast if that element has been inserted only shortly before. We omit the precise definitions here, since they are not needed for our proofs; see \cref{sec:wsprops} for more details.
	
	\section{Generalized topological heapsort}\label{sec:main}
	
	Let us now formally state the antimatroid-sorting meta-problem. Given is a set $\Sigma$, an antimatroid $A$ on $\Sigma$, and a comparison oracle for a permutation $\pi \in \perms(A)$. The task is to output $\pi$ (or, equivalently, output $\Sigma$ in sorted order according to the oracle comparisons $\prec_\pi$). Note the assumption $\pi \in \perms(A)$ in particular implies that $A$ is full.
	
	We call this a meta-problem since the precise representation of $A$ is unspecified. We discuss some possible representations in \cref{sec:app}. For now, we give a meta-algorithm for the problem (\cref{alg:ths}),
	using the characterization of $A$ as a monotone precedence system~$S$.

	\begin{algorithm}[H]
		\caption{Topological heapsort for antimatroids.}\label{alg:ths}
		\begin{algorithmic}
			\Statex \textbf{Input:} Set $\Sigma$, MPS $S = \{p_x\}_{x \in \Sigma}$, comparison oracle for some $\pi \in \perms(S)$.
			\State Initialize working-set priority queue $Q$ that contains each $x \in \Sigma$ with $p_x(\emptyset) = 1$
			\State $\alpha \gets \emptyword$
			\While{$Q$ is not empty}
				\State Delete the minimum from $Q$ and append it to $\alpha$
				\State For each $x \in \Sigma \setminus (\support{\alpha} \cup Q)$, add it to $Q$ if $p_x(\support{\alpha}) = 1$
			\EndWhile
			\State \Return $\alpha$
		\end{algorithmic}
	\end{algorithm}
	
	We can already show correctness of our meta-algorithm regardless of implementation details.
	
	\begin{lemma}
		Topological heapsort is correct for every input $(\Sigma, S, \pi)$.
	\end{lemma}
	\begin{proof}
		Let $S = \{p_x\}_{x \in \Sigma}$ and $\pi = x_1, x_2, \dots, x_n$.
		We need to show that topological heapsort produces $\alpha = \pi$.
		
		Since $\pi \in \perms(S)$, we have $p_{x_i}( \{x_1, x_2, \dots, x_{i-1}\} ) = 1$ for each $i \in [n]$. In particular, this means that $p_{x_1}(\emptyset) = 1$, so $x_1$ is contained in $Q$ at the start. Since also $x_1 \prec_\pi x_i$ for all $i > 1$, the first element added to $\alpha$ is indeed $x_1$.
		Suppose now that we have $\alpha = x_1 x_2 \dots x_{i-1}$ at the start of some iteration. Since $p_{x_i}(\support{\alpha}) = 1$, we have $x_i \in Q$. Since $x_i \prec_\pi x_j$ for all $j > i$, the algorithm next appends $x_i$ to $\alpha$. By induction, the algorithm indeed returns $\pi$.
	\end{proof}

	Two parts of the algorithm are left unspecified: How to identify the initial elements in~$Q$, and the elements inserted in each iteration. We now define an abstract data structure for this task. A \emph{candidate data structure} $C$ for a monotone precedence system $S = \{p_x\}_{x \in \Sigma}$ maintains a set $X \subseteq \Sigma$, initially $X = \emptyset$, and supports the following two operations:
	\begin{itemize}
		\item $C.\cdsinit()$: Set $X \gets \emptyset$ and report all $x \in \Sigma$ with $p_x(\emptyset) = 1$.
		\item $C.\cdsstep(x)$: Given $x \in \Sigma \setminus X$ with $p_x(X) = 1$, report all $y \in \Sigma$ with $p_y(X) = 0$ and $p_y(X \cup \{x\}) = 1$. Then, add $x$ to $X$.
	\end{itemize}
	
	Observe that $C$ allows at most $|\Sigma|$ calls to $\cdsstep$ after $\cdsinit$, one for each $x \in \Sigma$. Further, parameters of $\cdsstep$ calls form a sequence $x_1 x_2 \dots x_k \in \Lang(S)$. The \emph{total time} for $C$, written $\tctotal(C)$, is the overall time for one call to $\cdsinit$ and a valid sequence of $\cdsstep$ calls, in the worst case.
	
	Clearly, any correct candidate data structure can be used to finish the implementation of topological heapsort. The actual implementation will heavily depend on the given representation of the antimatroid; see \cref{sec:app}.
	
	\subsection{Running time analysis}
	
	Since we want our analysis to be independent of the candidate data structure, we split the work done by topological heapsort into two essential parts.
	\begin{itemize}
		\item The total time $\tctotal(C)$ spend by the candidate data structure, as defined above.
		\item The time required to initialize $Q$ (given the initial elements), delete the minimum from~$Q$, and insert given elements into~$Q$ in each step. Call this the \emph{queue time} $\tq$.
	\end{itemize}
	
	Since the loop performs at most $n = |\Sigma|$ iterations, the running time of topological heapsort is $\tctotal(C) + \tq + \fO(n)$.
	
	Note that the queue time does \emph{not} depend on the representation of $S$. In the following, we show that $\tq$ is always $\fO(n + \log |\perms(S)|)$, by reducing to the special case where $S$ is a \emph{partial order}, which has been solved by Haeupler et al.~\cite{HaeuplerHladikEtAl2025a}. The number of comparisons can be as large as the queue time, which is not optimal if $\log(|\perms(S)|) \ll n$. In \cref{sec:comparisons}, we improve this bound to the optimal $\fO( \log |\perms(S)|)$ by modifying the algorithm.
	
	\subparagraph*{Partial orders.}
	
	Let $G$ be a directed acyclic graph with $n$ vertices and $m$ edges, and consider the set $\topos(G)$ of topological orderings of $G$. Clearly, $\topos(G)$ is the set of linear extensions of a partial order on $V(G)$.
	Thus, $\topos(G)$ is an also antimatroid in $V(G)$ (see \cref{ex:poset}).
	
	It is easy to implement a candidate data structure for $\topos(G)$ via the standard topological sorting algorithm~\cite{Kahn1962,Knuth1968}, with total running time $\fO(m+n)$.
	Haeupler et al.~\cite{HaeuplerHladikEtAl2025a} showed that the queue time of topological heapsort is optimal for partial orders:
	
	\begin{lemma}[Haeupler et al.~\cite{HaeuplerHladikEtAl2025a}]\label{p:poset-opt-time}
		Let $S$ be a monotone precedence system corresponding to a partial order on a set $\Sigma$, and let $\pi \in \perms(S)$ be a total order on $\Sigma$. Then, for topological heapsort with input $(\Sigma,S,\pi)$, we have $\tq \in \fO( |\Sigma| + \log |\perms(S)|)$.
	\end{lemma}
	
	\subparagraph*{General antimatroids.}
	We now generalize \cref{p:poset-opt-time} to antimatroids. We need the following definition. The \emph{transcript} of a run of topological heapsort on an input $(\Sigma,S,\pi)$ is the sequence $(Q_0, Q_1, \dots, Q_n)$, where $Q_0$ is the set of elements initially in the priority queue, and $Q_i$ for $i \in [n]$ is the set of elements in the priority queue after the $i$-th step. Observe that, if $\pi = x_1 x_2 \dots x_n$, then $x_i \in Q_{i-1}$, and $x_i \notin Q_j$ for each $j \ge i$. See \cref{fig:transcript} for an example.
	
	\begin{figure}
		\newcommand{\verteq}{\rotatebox{90}{$\,=$}}
		\centering
		\begin{tikzpicture}[y=4mm, x=15mm, anchor=base]
			\newcommand\queue[2]{\node at (#1,0) {$Q_#1$}; \node at (#1,-1.2) {$\verteq$}; \node at (#1,-2) {#2};}
			\newcommand\x[2]{\node at (#1,0) {$x_#1$}; \node at (#1,.8) {$\verteq$}; \node at (#1,2) {#2};}
			\queue{0}{$\{a,c\}$}
			\queue{1}{$\{b,c\}$}
			\queue{2}{$\{c\}$}
			\queue{3}{$\{d,e\}$}
			\queue{4}{$\{e\}$}
			\queue{5}{$\emptyset$}
			
			\begin{scope}[yshift=8mm,shift={(-.5,0)}]
				\x{1}{$a$}
				\x{2}{$b$}
				\x{3}{$c$}
				\x{4}{$d$}
				\x{5}{$e$}
			\end{scope}
		\end{tikzpicture}
		\caption{An example transcript of topological heapsort for some antimatroid on $\Sigma = \{a,b,c,d,e\}$.}\label{fig:transcript}
	\end{figure}
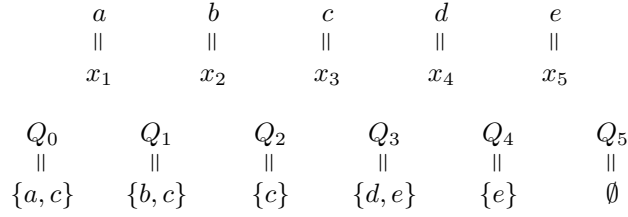
	
	A crucial fact that we use in the proof below is that the queue time $\tq$ \emph{only} depends on the transcript of the run. In particular, two runs on inputs $(\Sigma,S,\pi)$ and $(\Sigma,S',\pi)$ that have the same transcript also have the same queue time, even if $S \neq S'$.
	
	\begin{theorem}\label{p:opt-queue-time}
		Let $S$ be a monotone precedence system on $\Sigma$, and let $\pi \in \perms(S)$ be a total order on $\Sigma$. Then, for topological heapsort with input $(\Sigma,S,\pi)$, we have $\tq \in \fO( |\Sigma| + \log |\perms(S)|$.
	\end{theorem}
	\begin{proof}
		Consider a run of topological heapsort on $(\Sigma,S,\pi)$ with transcript $(Q_0, Q_1, \dots, Q_n)$, and let $\pi = x_1 x_2 \dots x_n$.
		
		We now define a partial order $P$ on $\Sigma$. For each $i \in [n]$, let $x_i \prec_P y$ for each $y \in \Sigma \setminus (Q_0 \cup Q_1 \cup \dots Q_{i-1})$. Informally, we have $x \prec_P y$ if $x$ is deleted from the priority queue before $y$ is inserted.
		Let $S_P$ be the monotone precedence system corresponding to $P$.
		Clearly, $\pi \in \perms(S_P)$. We claim the following:
		\begin{enumerate}[(i)]
			\item The transcript of running topological heapsort on $(\Sigma,S_P,\pi)$ is again $(Q_0, Q_1, \dots, Q_n)$.\label{claim:poset-equiv-trans}
			\item $\perms(S_P) \subseteq \perms(S)$.\label{claim:poset-sub}
		\end{enumerate}
		
		We first explain how this implies the theorem. By \cref{p:poset-opt-time}, running topological heapsort on $(\Sigma,S_P,\pi)$ has queue time $\fO( |\Sigma| + \log |\perms(S_P)| )$. By claim (\ref{claim:poset-equiv-trans}), running topological heapsort on $(\Sigma,S,\pi)$ has the same queue time (since the transcript is the same), which, by claim (\ref{claim:poset-sub}), is at most $\fO( |\Sigma| + \log |\perms(S)| )$, as required.
		
		We now prove claim (\ref{claim:poset-equiv-trans}). Let $(Q_0', Q_1', \dots, Q_n')$ be the transcript of running topological heapsort on $(\Sigma,S_P, \pi)$. Each $y \in Q_0$ is a minimal element of $P$ by definition, so $Q_0 \subseteq Q_0'$. On the other hand, each $y \notin Q_0$ satisfies $x_1 \prec_P y$ by definition and thus $y \notin Q_0'$, so $Q_0 = Q_0'$.

		Now take some $j \in [n]$ and assume that by induction, we have $Q_i = Q_i'$ for all $i < j$. For each $y \in Q_j$, we have $\{ z \in \Sigma \mid z \prec_P y \} \subseteq \{x_1, x_2, \dots, x_j\}$, and $y \notin \{x_1, x_2, \dots, x_j\}$, so $y \in Q_j'$ by definition.
		On the other hand, each $y \notin Q_j$ satisfies either (a) $y \in \{x_1, x_2, \dots, x_j\}$ or (b) $y \notin Q_0 \cup Q_1 \cup \dots \cup Q_j$. (a) directly implies $y \notin Q'_j$ ($y$ is removed earlier). (b) implies $x_{j+1} \prec_P y$, which again implies $y \notin Q'_j$.

		We finish the proof with claim (\ref{claim:poset-sub}). Let $\alpha = x_1' x_2' \dots x_n' \in \perms(S_P)$. We show that $\alpha \in \perms(S)$, i.e., that for each $i \in [n]$, we have $p_{x_i'}(\{x_1', x_2', \dots, x_{i-1}'\}) = 1$.
		
		Let $y \in \Sigma$, and let $k$ be minimal such that $y \in Q_k$. Let $X_k = \{x_1, x_2, \dots, x_k\}$. First, by definition of $Q_k$, we have $p_y(X_k) = 1$. Second, we have $\{x \in \Sigma \mid x \prec_P y\} = X_k$ by minimality of $k$.
		Now suppose $y = x'_i$, and let $X'_{i-1} = \{x'_1, x'_2, \dots, x'_{i-1}\}$. Then $p_y'(X'_{i-1}) = 1$, implying $X'_{i-1} \supseteq X_k$. By monotonicity, we have $p_y(X'_{i-1}) \ge p_y(X_k) = 1$, as desired.
	\end{proof}
		
	\subparagraph{Remark.}
	The technique of constructing an auxiliary partial order is inspired by the work of Haeupler, Hladík, Rozhoň, Tarjan, and Tětek~\cite{HaeuplerHlad'ikEtAl2024} (we discuss further connections in \cref{sec:dist-order}).
	In fact, the exact partial order $P$ in the proof of \cref{p:opt-queue-time} appears in the revised version\footnote{An earlier preprint~\cite{HaeuplerHladikEtAl2024} used a different proof technique.} of the paper that introduced topological heapsort~\cite{HaeuplerHladikEtAl2025a}. They observe that $P$ is an \emph{interval order}, which means $|\perms(P)|$ is relatively easy to bound and can be related to the working-set bound.
	
	The algorithm of van der Hoog, Rotenberg, and Rutschmann~\cite{vanderHoogRotenbergEtAl2025} is also analyzed with the help of interval orders, but seemingly cannot be easily generalized to antimatroids.
	
	\section{Antimatroid representations}\label{sec:app}

	In this section, we consider several representations of antimatroids (some only covering special cases), and discuss how to implement candidate data structures (CDSs) for these representations. If $C$ is a CDS for an animatroid $A$ with some representation, then we call $C$ \emph{efficient} if $\tctotal(C) \in \fO(n)$, where $n$ is the size of the representation (in machine words). \Cref{p:opt-queue-time} implies that topological heapsort with an efficient CDS has optimal running time.

	\subsection{Precedence formulas}\label{sec:app:formulae}
	
	We first consider an explicit representation of general full antimatroids based on monotone precedence systems. The \emph{precedence formula representation} of a full antimatroid $A$ on $\Sigma$ is a collection $\{F_x\}_{x \in \Sigma}$, where $F_x$ is a monotone boolean formula on the variable set $V_x = \Sigma \setminus \{x\}$. (A \emph{monotone} boolean formula allows only the operators $\vee$ and $\wedge$ and the constants 0 and 1, notably prohibiting the negation $\neg$.)
	
	Each formula $F_x$ represents a precedence function $p_x$ as follows: For $Y \subseteq \Sigma$, we let $p_x(Y)$ be the value of $F_x$ when setting each $y \in Y$ to~1, and each $y \in X \setminus (Y \cup \{x\})$ to~0.
	Since any monotone function $f$ can be represented as a monotone boolean formula, this representation is suitable for every antimatroid. In the following, we assume each formula is given as a linked list of symbols with each variable, operator, constant, or parenthesis in the formula occupying one machine word. Other representations, like parsing trees and circuits, are also possible.
	
	An CDS $C$ for $\{F_x\}_{x \in \Sigma}$ is implemented as follows. We need two preprocessing steps. First, for each formula, we simplify it. We repeatedly replace substrings $(1 \wedge x) \rightarrow x$, $(1 \vee x) \rightarrow 1$, $(0 \wedge x) \rightarrow 0$, and so on, in linear overall time. Afterwards, we have $F_x = 1$ for each $x$ with $p_x(\emptyset) = 1$.

	Second, compute a directed graph $G$, where $V(G) = \Sigma$, and $(x,y) \in E(G)$ if $F_y$ contains the variable $x$. For each such edge, we store a list of pointers to each occurrence of $x$ within~$F_y$. Note that the size of $G$ (including pointers) is at most the total size of all formulas.
	
	The initial candidate set reported by $\cdsinit()$ is the set of sources in $G$. $\cdsstep(x)$ is implemented as follows: For each out-neighbor $y$ of $x$, replace each occurrence of $x$ in $F_y$ with 1. Then, simplify $F_y$ as outlined above, and report $y$ if $F_y = 1$ after the simplification. Finally, remove $x$ from the graph.
	To see that the data structure is correct, observe that after calling $\cdsstep(x)$ on each $x$ in some set $X \subset \Sigma$, we have $F_y = 1$ for each $y$ with $p_y(X) = 1$.
	
	The total running time is linear, since each edge is visited once, and simplification is linear overall. Thus, the data structure is efficient, and so:
	\begin{theorem}
		Given an antimatroid $A$ on $\Sigma$ in a precedence formula representation of length~$n$, we can solve the $\perms(A)$-sorting problem in optimal time $\fO( n + \log |\perms(A)| )$.
	\end{theorem}
	
	There are two standard representations of antimatroids that are similar to precedence formulas: \emph{Rooted set collections} and \emph{alternative precedence structures}~\cite{KorteSchraderEtAl1991}. These roughly correspond to precedence formulas in conjuctive normal form and disjunctive normal form, respectively. A third related characterization is discussed in the next section.
	
	\subsection{Elementary ranking conditions}%
	
	\emph{Optimality theory}~\cite{PrinceSmolensky2004} is a framework in linguistics where a so-called \emph{grammar} is modeled by a \emph{ranking} (permutation) of so-called \emph{constraints}. Several different constraint rankings may yield the same grammar. In particular, each grammar can be described by a set of \emph{elementary ranking conditions (ERCs)}~\cite{MerchantRiggle2015}. Each ERC is a tuple $(A,B)$ of constraint sets, and a constraint ranking $\alpha$ describes a given grammar if and only if for each ERC $(A,B)$, there exists a constraint $a \in A$ that outranks (precedes in $\alpha$) all constraints in $B$.
	
	Let us formally say that a permutation $\pi$ of $\Sigma$ is \emph{consistent} with an ERC $(A,B)$, with $A, B \subseteq \Sigma$ and $A \cap B = \emptyset$, if there exists some $a \in A$ such that $a \prec_\pi b$ for all $b \in B$. For a set $R$ of ERCs, let $\perms(R)$ be the set of permutations of $\Sigma$ that are consistent with each ERC in~$R$.
	Merchant and Riggle~\cite{MerchantRiggle2015} showed that $\perms(R)$ is the set of permutations of an antimatroid, and each antimatroid can be described by a set of ERCs. We give a short proof based on our formalism and previous observations. We start with a technical lemma.
	
	\begin{lemma}\label{p:erc-equiv}
		Let $\Sigma$ be an alphabet, let $(A,B)$ be an ERC, and let $\pi$ be a permutation of $\Sigma$. Then
		$\forall b \in B : \exists a \in A : a \prec_\pi b \iff \exists a \in A : \forall b \in B : a \prec_\pi b$.
	\end{lemma}
	\begin{proof}
		Let $a$ be minimal w.r.t.\ $\prec_\pi$ such that $a \in A$. If there is some $b \in B$ with $b \prec_\pi a$, then both sides of the equation are false. On the other hand, if $a \prec_\pi b$ for all $b \in B$, then both sides are true.
	\end{proof}
	
	\begin{proposition}[Merchant and Riggle~\cite{MerchantRiggle2015}]\label{p:ercs}
		Let $\Sigma$ be an alphabet. For each set $R$ of ERCs, there exists an antimatroid $A$ on $\Sigma$ such that $\perms(R) = \perms(A)$, and for each full antimatroid $A$, there exists a set of ERCs $R$ such that $\perms(R) = \perms(A)$.
	\end{proposition}
	\begin{proof}
		Let $R$ be an ERC set. We have
		\begin{align*}
			\perms(R) & = \{ \pi \mid \forall (A,B) \in R: \exists a \in A: \forall b \in B: a \prec_\pi b \}\tag{def.}\\
			& = \{ \pi \mid \forall (A,B) \in R: \forall b \in B : \exists a \in A : a \prec_\pi b \}\tag{lem.~\ref{p:erc-equiv}}\\
			& = \{ \pi \mid \forall b \in \Sigma: \forall (A,B) \in R, b \in B: \exists a \in A : a \prec_\pi b \}
		\end{align*}
		
		Now define the MPS $S = \{p_x\}_{x \in \Sigma}$ such that $p_x(X) = 1$ if and only if
		\[ \forall (A,B) \in R, b \in B: \exists a \in A : a \in X. \]
		
		It is easy to check that each $p_x$ is indeed monotone, and that $\perms(R) = \perms(S)$. Since each full antimatroid can be characterized by an MPS (\cref{p:mps-char}), this finishes the proof.
	\end{proof}
	
	To construct a CDS for an ERC set, we can adapt the method used for topological orderings of graphs. For each element $x \in \Sigma$, maintain the number of ERCs $(A,B)$ with $x \in B$; an element becomes available when this number becomes zero. Whenever $\cdsstep(x)$ is called, remove each ERC $(A,B)$ with $x \in A$. Thus, given an ERC set $E$, we can solve the $\perms(E)$-sorting problem in optimal time.

	Precedence formulas and ERC sets are powerful enough to describe arbitrary antimatroids. In the following, we discuss two more representations that only apply to subclasses of antimatroids.
	
	\subsection{Vertex search and distance orderings}\label{sec:dist-order}

	Given a connected graph $G$ and a designated root $r \in V(G)$, the \emph{vertex search antimatroid} on $(G,r)$, written $\AMVS_{G,r}$, contains the empty word $\emptyword$ as well as each simple word $v_1 v_2 \dots v_k$ where $v_1 = r$ and each prefix $v_1 v_2 \dots v_i$ induces a connected subgraph of $G$. The words in $\AMVS_{G,r}$ model each way of traversing the graph in a manner similar to BFS or DFS. The basic idea is easily extended to directed graphs; here we require that each prefix $v_1 v_2 \dots v_i$ induces a subgraph of $G$ where every vertex is reachable from $r = v_1$. (We assume that each vertex is reachable from $r$ in $G$.)
	
	It is easy to see that $\AMVS_{G,r}$ is indeed an antimatroid;
	for example, as a precedence formula for each vertex $v \neq r$, take $F_v = u_1 \vee u_2 \vee \dots \vee u_k$, where $u_1, u_2, \dots, u_k$ are the neighbors (resp.\ in-neighbors) of $v$.
	However, not every antimatroid is a vertex search antimatroid.

	Again, the restricted sorting problem for vertex search antimatroids is solved in optimal time with the precedence formula representation given above.
	
	The vertex search antimatroid plays a role in the \emph{distance ordering} problem, a variant of single-source shortest path (SSSP) problem. Given is an edge-weighted directed graph $(G,w)$ with a designated root $r$, and the task is to order the vertices by distance from $r$. Recently, Haeupler et al.~\cite{HaeuplerHlad'ikEtAl2024} showed that Dijkstra's classical SSSP algorithm is \emph{universally optimal} when equipped with a certain working-set heap. Universally optimal means, informally, that for every fixed graph, the algorithm is worst-case optimal w.r.t.\ the weight function. In fact, they essentially match the information-theoretic lower bound with a running time of $\fO( \log |T| + |V(G)|+|E(G)|)$, where $T$ is the set of possible distance orderings for the input graph $G$ with root $r$, with any weight function.
	
	We now demonstrate a strong connection between antimatroid sorting and the distance ordering problem.\footnote{Such a connection is not surprising, since the technique used in the analysis of Dijkstra's algorithm~\cite{HaeuplerHlad'ikEtAl2024} inspired the original analysis of topological heapsort~\cite{HaeuplerHladikEtAl2025a} and large parts of this paper.}
	First, the set $T$ of possible distance orderings is precisely the vertex search antimatroid; we prove this in \cref{app:dijkstra}. Thus, we can run topological heapsort to find a distance ordering with optimal running time -- \emph{if} we are allowed to directly compare two vertices. That means comparing the distances of two vertices to $r$, which we cannot do in the distance ordering problem (without computing the distances first).

	Still, it can be seen that the \emph{transcript} (see the proof of \cref{p:opt-queue-time}) of a run of Dijsktra's algorithm on an input $(G,w,r)$ is the same as running topological heapsort on the corresponding vertex search antimatroid. This gives an alternative proof of universal optimality; see \cref{sec:dijkstra} for more details.

	\subsection{Elimination orderings}\label{sec:peo}
		
	An undirected graph is \emph{chordal} if has no induced cycles of length four or more. It is well-known that chordal graphs can be characterized by the existence of \emph{perfect elimination orderings (PEOs)}. A PEO is obtained by successively removing \emph{simplicial} vertices, that is, vertices whose neighborhood form a clique.
	
	The set of elimination orderings form an antimatroid, also known as the \emph{simplicial vertex pruning}~\cite{BjoernerZiegler1992} or \emph{simplicial shelling}~\cite{KorteSchraderEtAl1991} antimatroid. Given a graph $G$, we can define the simplicial vertex pruning antimatroid $\AMSVP_G$ via the MPS $\{p_v\}_{v \in V(G)}$ with
	\begin{align*}
		p_v(U) = 
		\begin{cases}
			1, &\text{ if $v$ is simplicial in $G \gminus U$;}\\
			0, &\text{ otherwise.}
		\end{cases}
	\end{align*}
	
	A candidate data structure for $A_G$ is essentially a data structure that maintains the set of simplicial vertices in $G$ under simplicial vertex deletion.
	
	There are known static and fully-dynamic algorithms to compute simplicial vertices in \emph{arbitrary} graphs~\cite{KloksKratschEtAl2000,LinSoulignacEtAl2012}. Since we only care about chordal graphs and only need vertex deletion, we can use a relatively simple data structure based on \emph{clique trees}. Details are found in \cref{sec:dec-simplicial}. The total running time of removing all vertices is $\fO(m+n)$. Thus, we can solve the corresponding restricted sorting problem in optimal time.
	
	It turns out that in this special case, we already achieve \emph{comparison} optimality, even without the improvements in \cref{sec:comparisons}. Since every chordal graph $G$ with $|V(G)| \ge 2$ has at least two simplicial vertices, we have $|\perms(\AMSVP_G)| \ge 2^{n-1}$. Thus, the $\fO(n + \log |\perms(\AMSVP_G)|)$ comparisons performed by topological heapsort (\cref{p:opt-queue-time}) are optimal.
	
	\section{Comparison optimality}\label{sec:comparisons}
	
	In this section, we modify the topological heapsort algorithm to only use an optimal $\fO(\log |\perms(S)|)$ comparisons, instead of $\fO(|\Sigma| + \log |\perms(S)|)$, while keeping the optimal running time. The special case of partial orders was again solved by Haeupler et al.~\cite{HaeuplerHladikEtAl2025a}. Our algorithm is similar, but some non-trivial adaptation is required.
	
	First of all, note that the additional $\fO(n)$ term is only relevant when $|\perms(S)| \le 2^{o(n)}$. In the partial order/DAG-sorting case, Haeupler et al.~\cite{HaeuplerHladikEtAl2025a} showed that this condition implies the existence of a long directed path in the DAG. Observe that such a path is pre-sorted. We prove a corresponding fact for antimatroids (\cref{sec:bottlenecks}).
	
	Haeupler et al.\ then proceed roughly as follows. They first compute a set of \emph{marked} vertices on the pre-sorted path. Intuitively, the marked vertices represent the interaction with the remainder of the graph: The removal of unmarked vertices does not affect ``availability'' of vertices outside of the path. This means that often, a large set of unmarked vertices can be removed at once, without any additional comparisons.
	
	This strategy cannot be adapted to general antimatroids in a straight-forward way, since the availability constraints are too complex to preprecompute a small set of marked vertices. We instead use the following three-step algorithm (given is an antimatroid $A$ on $\Sigma$):
	
	\begin{enumerate}
		\item Compute a long pre-sorted sequence $\beta$ (\cref{sec:bottlenecks}).
		\item Sort the set $\Gamma = \Sigma \setminus \support{\beta}$, using topological heapsort on the sub-antimatroid of $A$ \emph{induced} by $\Gamma$ (\cref{sec:sub-sort}). This gives us a sorted sequence $\gamma$.
		\item Merge the two sequences $\beta$ and $\gamma$, utilizing the antimatroid constraints (\cref{sec:merge}).
	\end{enumerate}
	
	All three steps can be implemented efficiently using a candidate data structure for $A$. We obtain:
	
	\begin{theorem}\label{p:main-comp}
		Given a set $\Sigma$, a candidate data structure $C$ for an antimatroid $A$ on $\sigma$, and an oracle for a total order $\pi \in \perms(A)$, we can sort $\Sigma$ w.r.t.\ $\prec_\pi$ in time $\fO( \tctotal(C) + \log |\perms(A)|)$ and with $\fO(\log |\perms(A)|)$ comparisons.
	\end{theorem}
	
	This means that each efficient candidate data structure presented in \cref{sec:app} yields a comparison-optimal sorting algorithm.
	
	\subsection{Bottlenecks in antimatroids}\label{sec:bottlenecks}
	
	Let $S = \{p_x\}_{x \in \Sigma}$ be a monotone precedence system on $\Sigma$. The \emph{layer sequence} of $S$ is a partition of $\Sigma$ into nonempty \emph{layers} $L_1, L_2, \dots, L_k$ inductively defined as follows:
	\begin{itemize}
		\item $L_1 = \{ x \in \Sigma \mid p_x(\emptyset) = 1 \}$.
		\item $L_i = \{ x \in \Sigma \setminus \bar{L}_{i-1} \mid p_x(\bar{L}_{i-1}) = 1 \}$, where $\bar{L}_{i-1} = L_1 \cup L_2 \cup \dots \cup L_{i-1}$, for each $2 \le i \le k$.
	\end{itemize}

	Observe that a layer sequence exists if and only if $\perms(S) \neq \emptyset$, and that a layer sequence is always unique.
	If a candidate data structure $C$ is available, we can compute the layer sequence in $\fO(\tctotal(C))$ time as follows: Call $C.\cdsinit()$ and add the reported elements to $L_1$. Then, call $C.\cdsstep(x)$ for each $x \in L_1$, add the reported elements to $L_2$, and so on.
	
	\begin{example}%
		If $\perms(S)$ is the set of permutations of a vertex search antimatroid $\AMVS_{G,s}$, then $L_i$ contains each vertex $v$ with distance $i-1$ from $s$.
	\end{example}
	
	\begin{example}%
		If $\perms(S)$ is the set of topological orderings of a graph $G$ with $V(G) = \Sigma$, then $L_i$ contains each vertex $v$ where the \emph{longest path} ending at $v$ has length $i$. This corresponds to the layering defined by Haeupler et al.~\cite{HaeuplerHladikEtAl2025a}.
	\end{example}

	\begin{restatable}{lemma}{restateLayerCount}\label{p:layer-count}
		 Let $S$ be a monotone precedence system with $k$ layers. Then $|\perms(S)| \ge 2^{n-k}$.
	\end{restatable}
	\begin{proof}
		We essentially follow Haeupler et al.~\cite{HaeuplerHladikEtAl2025a}.
		Let $L_1, L_2, \dots, L_k$ be the layer sequence for $S$. For each $i \in [k]$, let $\alpha_i$ be an arbitrary permutation of $L_i$. We claim that the word $\alpha = \alpha_1 \alpha_2 \dots \alpha_k$ is contained in $\perms(S)$. Indeed, if $x \in L_i$, then all $y \in \bar{L}_{i-1}$ precede $x$ in $\alpha$. Since $p_x(\bar{L}_{i-1}) = 1$ by definition of $L_i$, we have $\alpha \in \perms(S)$.
		
		There are $\prod_{i = 1}^k |L_i|!$ ways of constructing $\alpha$. Writing $n_i = |L_i|$ and $n = |\Sigma|$, we have
		\[
			|\perms(S)| \ge \prod_{i=1}^k n_i! \ge \prod_{i=1}^k 2^{n_i-1} = 2^{n-k}.\qedhere
		\]
	\end{proof}
		
	We now proceed to show how to find a long sequence $x_1 x_2 \dots x_k \in \swords{\Sigma}$ such that for all~$i \in [k-1]$, we have $x_i \prec_\alpha x_{i+1}$ for all $\alpha \in \perms(A)$; in other words, the sequence is pre-sorted. In the DAG-sorting case, we can simply take a longest path in the DAG. Here, we need a different approach, which is also related to the technique of Haeupler et al.~\cite{HaeuplerHladikEtAl2025a}.
	
	\begin{lemma}\label{p:layer-ordering}
		Let $L_1, L_2, \dots, L_k$ be the layers of a monotone precedence system $S$, and let $\pi \in \perms(S)$. Then, for all $1 \le i < j \le k$ and $y \in L_j$, there exists an $x \in L_i$ such that $x \prec_\pi y$.
	\end{lemma}
	\begin{proof}
		Suppose the lemma does not hold for some pair $i < j$, i.e., there exists some $y' \in L_j$ such that $y' \prec_\pi x$ for all $x \in L_i$.
		Let $y = \min_\pi(\Sigma \setminus \bar{L}_i) = \min_\pi(L_{i+1} \cup L_{i+2} \cup \dots \cup L_k)$. Since $y \preceq_\pi y'$, we have $y \prec_\pi x$ for all $x \in L_i$, implying that actually $y = \min_\pi(\Sigma \setminus \bar{L}_{i-1})$. This means that $p_y(\bar{L}_{i-1}) = 1$, implying that $y \in L_i$, a contradiction.
	\end{proof}

	An element $x \in \Sigma$ is a called a \emph{bottleneck} if it is the only element in its layer. Let $t$ be the number of bottlenecks and let $k$ be the number of layers. Since each non-bottleneck layer contains at least two vertices, we have $n-t \ge 2(k-t)$, which implies $n-k \ge \tfrac{n-t}{2}$. By \cref{p:layer-count}, we have
	\begin{corollary}\label{p:count-bottlenecks}
		Let $S$ be an MPS with $t$ bottleneck elements. Then $|\perms(S)| \ge 2^{(n-t)/2}$.
	\end{corollary}
	
	The \emph{bottleneck sequence} $\beta = b_1 b_2 \dots b_t$ of $S$ contains all bottlenecks, ordered by the layer they appear in. \Cref{p:layer-ordering} implies that, for every permutation $\pi \in \perms(S)$, we have $b_1 \prec_\pi b_2 \prec_\pi \dots \prec_\pi b_t$. Overall, we obtain
	\begin{lemma}\label{p:bot-seq}
		Let $S$ be an MPS on $\Sigma$, and write $n = |\Sigma|$. Then there exists a sequence $\beta = b_1 b_2 \dots b_t$ that appears as a subsequence in every $\pi \in \perms(S)$, and $|\perms(S)| \ge 2^{(n-t)/2}$.
		Moreover, given a candidate data structure $C$ for $S$, we can compute $\beta$ in time $\fO(\tctotal(C))$.
	\end{lemma}

	\subsection{Sorting subsets}\label{sec:sub-sort}
	
	\newcommand{\restr}[1]{|_{#1}}
	
	In this section, we show how to sort the non-bottleneck elements $\Sigma \setminus \support{\beta}$ of an antimatroid~$A$. Recall that a candidate data structure $C$ is given, and we want to sort in overall time $\fO( \tctotal(C) + \log |\perms(A)|)$, with $\fO(\log |\perms(A)|)$ comparisons (see \cref{p:main-comp}). We show how to sort any subset $\Gamma \subseteq \Sigma$ within this budget.
	
	We first need some definitions.
	Let $\Sigma$ be an alphabet. Given a word $\alpha$ and a subset $\Gamma \subseteq \Sigma$, the \emph{restriction} of $\alpha$ to $\Gamma$, written $\alpha\restr{\Gamma}$, is obtained by removing all letters in $\Sigma \setminus \Gamma$ from $\alpha$. For a simple language $L \subseteq \swords{\Sigma}$, let $L\restr{\Gamma} = \{ \alpha\restr{\Gamma} \mid \alpha \in L\}$.
	It is known that if $A \subseteq \swords{\Sigma}$ is an antimatroid, then the restriction $A\restr{\Gamma}$ is also an antimatroid (called the \emph{trace}~\cite{KorteSchraderEtAl1991,BjoernerZiegler1992}).
	
	Let us now assume we have a candidate data structure $C_\Gamma$ for $A\restr{\Gamma}$. Then we can use topological heapsort to sort $\Gamma$. The running time is $\tctotal(C_\Gamma) + \fO(\log |\perms(A\restr{\Gamma})|)$, and we need $\fO(|\Gamma| + \log |\perms(A\restr{\Gamma})|)$ comparisons. First, observe that $|\perms(A\restr{\Gamma})| \le |\perms(A)|$, since the restriction operation surjectively maps $\perms(A\restr{\Gamma})$ to $\perms(A)$. Second, recall that $|\Gamma| \le 2 \log |\perms(A)|$ if $\Gamma$ is obtained by removing all bottleneck elements (\cref{p:count-bottlenecks}). If we further assume that $\tctotal(C_\Gamma) \le \fO(\tctotal(C))$, then we are within the budget of \cref{p:main-comp}.
	
	It remains to build a candidate data structure $C_\Gamma$ for $A\restr{\Gamma}$ with $\tctotal(C_\Gamma) \le \fO(\tctotal(C))$. Suppose $C$ is given. Then $C_\Gamma.\cdsinit()$ is implemented as follows. First, call $C.\cdsinit()$, and store the result in a set $Y$.
	Then, we perform the following \emph{cleanup} procedure: As long as there is some $y \in Y \setminus \Gamma$, we remove $y$ from $Y$, call $C.\cdsstep(y)$, and then add the newly reported elements to $Y$. Finally, when $Y \subseteq \Gamma$, we report $Y$.

	To implement $C_\Gamma.\cdsstep(y)$, we similarly first call $C.\cdsstep(y)$, store the result in a set $Y$, and then cleanup as above.
	
	The total running time of $C_\Gamma$ is clearly $\tctotal(C) + \fO(|\Sigma|) = \fO(\tctotal(C))$. We now show that $C_\Gamma$ is correct, starting with a formalization of correctness of a candidate data structure~$C$.
	
	First, for $\alpha = a_1 a_2 \dots a_k \in \swords{\Sigma}$, let $R$ be the overall set of elements reported when calling $C.\cdsinit()$ and then $C.\cdsstep(a_i)$ for $i = 1, 2, \dots, k$, and let $\Av(C,\alpha) = R \setminus \support{\alpha}$.
	(For simplicity, we assume here that $C.\cdsstep(x)$ may be called even when $x$ is not currently available, and may then behave arbitrarily.)
	Second, for an antimatroid $A$ and some word $\alpha \in A$, define $\Cont(A, \alpha) = \{ x \in \Sigma \mid \alpha x \in A \}$. Observe that if $\{p_x \mid x \in \Sigma\}$ is a monotone precedence system for $A$, then $\Cont(A, \alpha) = \{ x \in \Sigma \setminus \support{\alpha} \mid p_x(\support{\alpha}) = 1\}$
	
	\begin{observation}
		$C$ is a candidate data structure for $A$ if and only if $\Av(C, \alpha) = \Cont(A,\alpha)$ for all $\alpha \in A$.
	\end{observation}
	
	We now show that $C_\Gamma$ is a candidate data structure for $A\restr{\Gamma}$. Consider a word $\alpha = a_1 a_2 \dots a_k$, and suppose we call $C_\Gamma.\cdsinit()$ and then $C_\Gamma.\cdsstep(a_i)$ for $i = 1, 2, \dots, k$. For each $i \in \{0, 1, \dots, k\}$, the $i$-th call of $\cdsstep$ (or $\cdsinit$, if $i = 0$) makes an additional sequence of calls $C.\cdsstep(x)$; say these calls are with parameters $b_{i,1}, b_{i,2}, \dots, b_{i,j_i}$, let $\beta_i = b_{i,1} b_{i,2} \dots b_{i,j_i}$, and let $\beta = \beta_0 a_1 \beta_1 a_2 \beta_2 \dots a_k \beta_k$. Since $C_\Gamma$ reports precisely those elements reported by $C$ that are contained in $\Gamma$, we have $\Av(C_\Gamma, \alpha) = \Av(C, \beta) \cap \Gamma$.
	
	Since $C$ is a CDS for $A$, we have $\Av(C, \beta) \cap \Gamma = \Cont(A, \beta) \cap \Gamma$. Take some $x \in \Cont(A, \beta) \cap \Gamma$. Then $\beta x \in A$, which implies that $(\beta x)\restr{\Gamma} = (\beta\restr{\Gamma}) x = \alpha x \in A_\Gamma$, and thus $x \in \Cont(A\restr{\Gamma}, \alpha)$. Thus, we have shown that $\Av(C_\Gamma, \alpha) \subseteq \Cont(A\restr{\Gamma},\alpha)$, which is one direction of the correctness proof.
	
	We now show $\Av(C_\Gamma, \alpha) \supseteq \Cont(A\restr{\Gamma},\alpha)$. Let $x \in \Cont(A\restr{\Gamma},\alpha)$. Then $\alpha x \in A\restr{\Gamma}$, implying that there is some $\beta^* \in A$ with $\beta^*\restr{\Gamma} = \alpha x$. Assume that $\beta^*$ is has minimum length, so $\beta^* = \beta'x$ for some $\beta'$. Since $\beta'\restr{\Gamma} = \alpha$, we can decompose $\beta'$ as follows: $\beta' = \beta'_0 a_1 \beta'_1 a_2 \beta'_2 \dots a_k \beta'_k$.
	
	We claim that $\support{\beta}'_i \subseteq \support{\beta}_i$ for all $i \in \{0, 1, \dots, k\}$.
	Recall that $\beta' x \in A$, so $p_x(\support{\beta}') = 1$, and our claim implies that $p_x(\support{\beta}) = 1$ by monotonicity.
	This, in turn, implies that $\beta x \in A$, and thus $x \in \Av(C_\Gamma,\alpha)$ by construction of $\beta$. This shows that $\Av(C_\Gamma, \alpha) \supseteq \Cont(A\restr{\Gamma},\alpha)$.
	
	We finish by proving the claim. Consider the first letter $b'_{0,0}$ in $\beta'_0$. Clearly, we have $p_{b'_{0,0}}(\emptyset) = 1$, so $b'_{0,0}$ is reported by $C.\cdsinit()$. Since $b'_{0,0} \notin \Gamma$, we know that $C_\Gamma.\cdsinit()$ calls $C.\cdsstep(b'_{0,0})$ at some point, which means that $b'_{0,0} \in \support{\beta_0}$ by definition.
	
	For the second letter $b'_{0,1}$, we have $p_{b'_{0,1}}(\{b'_{0,0}\}) = 1$. Recall that $C_\Gamma.\cdsinit()$ calls $C.\cdsstep(b'_{0,0})$ in the cleanup step, which then reports $b'_{0,1}$ (if it was not reported before). Again, since $b'_{0,1} \notin \Gamma$, we know that $C.\cdsstep(b'_{0,1})$ is called, implying $b'_{0,1} \in \support{\beta_0}$. By induction, this holds for every letter in $\beta'_0$, and another induction extends this to $\beta'_i$ with $i = 1, 2, \dots, k$. This concludes the correctness proof. We have shown:
	
	\begin{lemma}\label{p:sort-subset}
		Given a set $\Sigma$, a subset $\Gamma \subseteq \Sigma$, a candidate data structure $C$ for an antimatroid $A$ on $\Sigma$, and an oracle for a total order $\pi \in \perms(A)$, we can sort $\Gamma$ w.r.t.\ $\prec_\pi$ in time $\fO( \tctotal(C) + \log |\perms(A)|)$ and with $\fO(|\Gamma| + \log |\perms(A)|)$ comparisons.
	\end{lemma}
	
	\subsection{Merging}\label{sec:merge}
	
	We now want to optimally merge under antimatroid constraints:
	\begin{restatable}{theorem}{restateMerge}\label{p:merge}
		Let $A$ be an antimatroid on $\Sigma$, let $\pi \in \perms(A)$, and let $\Gamma, \Delta$ be a partition of~$\Sigma$. Let $\gamma = \pi\restr{\Gamma}$ and $\delta = \pi\restr{\Delta}$. Let $R \subseteq \perms(A)$ contain each permutation $\rho \in \perms(A)$ such that $\rho\restr{\Gamma} = \gamma$ and $\rho\restr{\Delta} = \delta$.
				
		Given $\Sigma$, $\gamma$, $\delta$, an oracle for $\pi$, and a candidate data structure $C$ for $A$, we can compute $\pi$ in time $\fO( \tctotal(C) + \log |R| )$ and with $\fO( |\Delta| + \log |R| )$ comparisons.
	\end{restatable}
	
	 \Cref{alg:merge} gives an implementation of \cref{p:merge}. It can be seen as a simplified version of Haeupler et al.'s \emph{heapsort with lookahead} algorithm~\cite{HaeuplerHladikEtAl2025a}.\footnote{In fact, if the antimatroid $A$ is given as a collection of precedence formulas, it simplifies to a \emph{partial order} under the assumption that $\gamma$ and $\delta$ are sorted. However, our proof works with general candidate data structures.} The function $\Call{exp-search}{d_j; g_ig_{i+1} \dots g_k}$ returns the smallest index $i^*$ with $i \le i^* \le k$ such that $d_j \prec_\pi g_{i^*}$, or $i^* = k+1$ if $g_k \prec_\pi d_j$. Implemented with an \emph{exponential search}~\cite{BentleyYao1976}, its time and comparison complexity is both $\fO( 1 + \log (i^*-i) )$.
	 
	 \begin{algorithm}
	 	\caption{Merging two pre-sorted sequences $\gamma$ and $\delta$ in an antimatroid, given as a candidate data structure $C$ and a comparison oracle for a total order $\pi$.}\label{alg:merge}
	 	\begin{algorithmic}
	 		\Procedure{merge}{$C, \gamma = g_1 g_2 \dots g_k, \delta = d_1 d_2 \dots d_\ell$}
	 		\State $X \gets C.\cdsinit()$\Comment{Currently available elements.}
	 		\State $i = 1, j = 1$
	 		\While{$i \le k$ or $j \le \ell$}
	 		\If{$j \le \ell$ and $d_j \in X$}
	 		\State $i^* = \Call{exp-search}{d_j; g_i g_{i+1} \dots g_k}$
	 		\ForEach{$i'$ with $i \le i' < i^*$}
	 		\State Report $g_{i'}$, $C.\cdsstep(g_{i'})$
	 		\EndFor
	 		\State Report $d_j$, $C.\cdsstep(d_j)$
	 		\State $i \gets i^*$
	 		\State Increment $j$
	 		\Else
	 		\State Report $g_i$, $C.\cdsstep(g_i)$
	 		\State Increment $i$
	 		\EndIf
	 		\EndWhile
	 		\EndProcedure
	 	\end{algorithmic}
	 \end{algorithm}
	 
	 Correctness of \cref{alg:merge} is easy to see; the algorithm essentially repeatedly outputs the smaller of the two elements $g_i$ and $d_j$ (although it does not always explicitly compare the two). The only non-obvious case is when $d_j \notin X$, but $g_i \in X$. But then we must have $g_i \prec_\pi d_j$ due to the antimatroid constraints, so we can safely report $g_i$ without any comparisons.
	 
	 The running time of \cref{alg:merge} is clearly $\tctotal(C) + \fO(n)$, where $n = |\Sigma| = |\gamma| + |\delta|$. We now analyze the number of comparisons.
	 
	 Consider a single loop iteration in \cref{alg:merge}. In the \emph{else} branch, no comparisons are made, so we can focus on the \emph{if} branch, with $\fO( 1 + \log (i^*-i) )$ comparisons. For each $j$, let $i_j$ and $i^*_j$ be the value of the respective variables in the iteration where $d_j$ is reported (note that this iteration is unique). Clearly, we have $i_j \le i^*_j$ for all $j \in [\ell]$, and $i^*_j \le i_{j+1}$ for each $j \in [\ell-1]$. The overall number of comparisons is
	 \begin{align}
	 	\fO\left( \sum_{j=1}^\ell 1 + \log (i_j^*-i_j) \right).\label{eq:merge:comps}
	 \end{align}
	 
	 The proof idea is that each $d_j$ could be inserted into any place in $g_{i_j} g_{i_j+1} \dots g_{i^*_j-1}$, which gives us enough different permutations so that the information-theoretic bound matches the running time of the exponential searches.
	 
	 More formally, for $j \in [\ell]$, let $M_j$ be the set of words formed by taking $g_{i_j}g_{i_j+1} \dots g_{i^*_j-1}$ and inserting $d_j$ in any of the $i^*_j - i_j + 1$ possible positions (including start and end). Let $\mu'_0 = g_1 g_2 \dots g_{i_1-1}$, let $\mu'_j = g_{i^*_j} g_{i^*_j+1} \dots g_{i_{j+1}-1}$ for $j \in [\ell-1]$ and let $\mu'_\ell = g_{i^*_\ell} \dots g_k$ (note that any $\mu'_j$ may be empty). Finally, let $U = \{ \mu'_0 \mu_1 \mu'_1 \mu_2 \mu'_2 \dots \mu_\ell \mu'_\ell \mid \mu_j \in M_j \text{ for } j \in \ell \}$. See \cref{fig:merge-analysis}.
	 
	 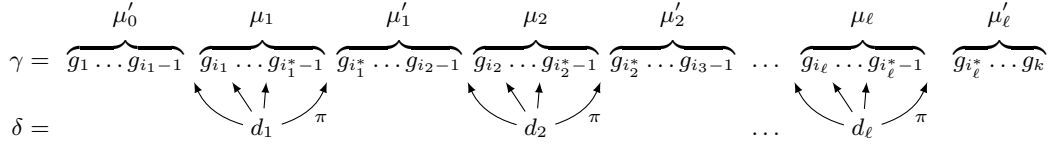
\begin{figure}
	 	\small\centering
	 	\begin{tikzpicture}[x=18mm, anchor=base, -latex]
	 		\node (dummy) at (-13.3mm,0) {};
	 		\node at (-10mm,0) {$\mathllap{\gamma =}$};
	 		\node at (-10mm,-9mm) {$\mathllap{\delta =}$};
	 		\begin{scope}[shift={(0,0)}]
	 			\node at (0,6mm) {$\mu_0'$};
	 			\node at (0,0) {$\overbrace{g_1 \dots g_{i_1-1}}$};
	 		\end{scope}
	 		\newcommand\mugroup[1]{
	 			\node at (0,6mm) {$\mu_{#1}$};
	 			\node at (0,0) {$\overbrace{g_{i_{#1}} \dots g_{i^*_{#1}-1}}$};
	 			\node[circle, inner sep=1pt] (d) at (0,-9mm) {$d_{#1}$};
	 			\draw (d) edge[bend left=25] (-9mm,-2mm);
	 			\draw (d) -- (-4mm,-2mm);
	 			\draw (d) -- (.5mm,-2mm);
	 			\draw (d) edge[bend right=25] node[below right=-.5mm] {\scriptsize$\pi$} (8.5mm,-2mm);
	 		}
	 		\newcommand\muprimegroup[2]{
	 			\node at (0,6mm) {$\mu_{#1}'$};
	 			\node at (0,0) {$\overbrace{g_{i^*_{#1}} \dots g_{i_{#2}-1}}$};
	 		}
	 		\begin{scope}[shift={(1,0)}]
	 			\mugroup{1}
	 		\end{scope}
	 		\begin{scope}[shift={(2,0)}]
	 			\muprimegroup{1}{2}
	 		\end{scope}
	 		\begin{scope}[shift={(3,0)}]
	 			\mugroup{2}
	 		\end{scope}
	 		\begin{scope}[shift={(4,0)}]
	 			\muprimegroup{2}{3}
	 		\end{scope}
	 		\def\xshift{3.5mm}
	 		\node[xshift=\xshift] at (4.5,0) {$\dots$};
	 		\node[xshift=\xshift] at (4.5,-9mm) {$\dots$};
	 		\begin{scope}[shift={(5,0)}, xshift={2*\xshift}]
	 			\mugroup{\ell}
	 		\end{scope}
	 		\begin{scope}[shift={(6,0)}, xshift={2*\xshift}]
	 			\node at (0,6mm) {$\mu_\ell'$};
	 			\node at (0,0) {$\overbrace{g_{i^*_\ell} \dots g_k}$};
	 		\end{scope}
	 	\end{tikzpicture}
	 	\caption{Visualization of permutations in the set $U$. The permutation $\pi$ is obtained by inserting each $d_j$ in the rightmost possible position.}\label{fig:merge-analysis}
	 \end{figure}
	 
	 \begin{lemma}\label{p:merge:u-r}
	 	$U \subseteq R$.
	 \end{lemma}
	 \begin{proof}
	 	First observe that for each $\alpha \in U$, we have $\alpha\restr{\Gamma} = \gamma$ and $\alpha\restr{\Delta} = \delta$. Thus, it suffices to show $U \subseteq \perms(A)$. Let $S = \{p_x\}_{x \in \Sigma}$ be a monotone precedence system for $A$, and let $\Gamma_i = \{ g_1, g_2, \dots, g_i \}$ and $\Delta_j = \{ d_1, d_2, \dots, d_j \}$ for $i \in [k]$ and $j \in [\ell]$.
	 	
	 	Let $\alpha \in U$, and let $x \in \Sigma$. Let $E_\alpha$, $E_\pi$ be the set of predecessors of $x$ in $\alpha$, resp.~$\pi$. We claim that $p_{x}(E_\alpha) = 1$. This implies $\alpha \in \perms(S) = \perms(A)$, as required.
	 	
	 	Towards our claim, consider first the case $x = d_j$ for some $j \in [\ell]$. Then $E_\alpha \supseteq \Gamma_{i_j-1} \cup \Delta_{j-1}$. Since $d_j \in X$ at the time the algorithm processes $d_j$, we clearly have $p_{d_j}( \Gamma_{i_j-1} \cup \Delta_{j-1} ) = 1$, implying $p_{d_j}(E_\alpha) = 1$ by monotonicity.
	 	
	 	Now consider $x = g_i$ for some $i \in [k]$. Then $E_\alpha \supseteq E_\pi$, since $\alpha$ only differs from $\pi$ by moving some elements $d_j$ to the left. Since $\pi \in \perms(A)$, we have $p_{g_i}(E_\pi) = 1$, which implies $p_{g_i}(E_\alpha) = 1$ by monotonicity.
	 \end{proof}
	 
	 We are now ready to finish the proof of \cref{p:merge}.
	 
	 \begin{lemma}
	 	\Cref{alg:merge} performs $\fO( \ell + \log |R| )$ comparisons.
	 \end{lemma}
	 \begin{proof}
	 	Observe that $|U| = \prod_{j=1}^\ell |M_j| = \prod_{j=1}^\ell (i_j^*-i_j+1) \ge \prod_{j=1}^\ell (i_j^*-i_j)$. By \cref{p:merge:u-r}, this implies that $\log |R| \ge \sum_{j=1}^\ell \log(i_j^*-i_j)$. From \cref{eq:merge:comps}, we get that the number of comparisons is $\fO( \ell + \log |R| )$.
	 \end{proof}
	
	\subsection{Putting things together}
	
	We now prove \cref{p:main-comp}. First, we compute the bottleneck sequence $\beta$ with \cref{p:bot-seq}. Then, we use \cref{p:sort-subset} to sort the subset $\Gamma = \Sigma \setminus \support{\beta}$. In other words, if $\pi$ is the underlying total order, we compute $\gamma = \pi\restr{\Gamma}$. Finally, we merge $\beta$ and $\gamma$ using \cref{p:merge}. The overall running time is $\fO( \tctotal(C) + \log |\perms(A)| + \log |R| )$, and we use $\fO(\log |\perms(A)| + |\Gamma| + \log |R|)$ comparisons, where $R$ is as defined in \cref{p:merge}. Note that $R \subseteq \perms(A)$. Moreover, we have $|\Gamma| = n - |\beta| \le 2 \log |\perms(A)|$ by \cref{p:bot-seq}. Thus, we need $\fO(\tctotal(C) + \log |\perms(A)|)$ time and $\fO(\log |\perms(A)|)$ comparisons, as desired.

	\section{Limitations and generalizations}\label{sec:generalizations}
	
	It was mentioned in the introduction that topological heapsort is correct for every set of total orders (given as the permutations of an accessible simple language), if a candidate data structure is available. In this section, we exhibit some interesting generalizations of antimatroids where topological heapsort is \emph{not optimal}.
	
	\subsection{Greedoids}
	
	\newcommand{\rot}[1]{#1^{\mathrm{r}}}
	
	Greedoids where first defined by Korte and Lovász~\cite{KorteLovasz1984a} and are generalizations of both matroids and antimatroids. There are two common definitions, one using set systems, and one using simple languages. We give the latter here, for obvious reasons.
	A \emph{greedoid} on an alphabet $\Sigma$ is a simple language $G \subseteq \swords{\Sigma}$, such that the following two axioms hold:
	
	\begin{enumerate}[(i)]
		\item For all $\alpha \in \swords{\Sigma}$ and $x \in \Sigma \setminus \support{\alpha}$, if $\alpha x \in A$, then $\alpha \in A$.\label{prop:gr:acc}
		\item For each $\alpha, \beta \in A$ with $|\alpha| > |\beta|$, there is some $x \in \support{\alpha} \setminus \support{\beta}$ such that $\beta x \in A$.\label{prop:gr:exchange}
	\end{enumerate}
	
	(\ref{prop:gr:acc}) is the usual accessibility axiom. (\ref{prop:gr:exchange}) is a weakening of the second antimatroid axiom: The condition $\support{\alpha} \not\subseteq \support{\beta}$ is replaced by $|\alpha| > |\beta|$.
	
	Consider the following language. Let $\Sigma = \{a_1, a_2, \dots, a_n\}$, and let $\alpha = a_1 a_2 \dots a_n$. For any word $\beta = b_1 b_2 \dots b_k$, define the \emph{rotation} $\rot{\beta} = b_k b_1 b_2 \dots b_{k-1}$. Define $T = \{ \rot{\alpha_1} \rot{\alpha_2} \dots \rot{\alpha_k} \mid \alpha_1 \alpha_2 \dots \alpha_k = \alpha \}$. In words, $T$ contains each permutation obtained by choosing an arbitrary set of non-overlapping subwords of $\alpha$ and rotating them. Note that we can obtain $\alpha$ itself this way by choosing each $\alpha_i$ to have length 1. Let $L$ be the set of all prefixes of words in $P$.
	
	Observe that topological heapsort sorts $P = \perms(L)$ with $\Theta(n \log n)$ comparisons. Indeed, all letters are available at the start and thus are immediately inserted into the queue. Even with the working set property, extracting all elements requires $\Omega(\log n)$ comparisons for a constant fraction of the elements. We now show that $\log |T| \in \Theta(n)$, implying that topological heapsort is not optimal.
	
	\begin{proposition}
		$|T| = 2^{n-1}$.
	\end{proposition}
	\begin{proof}
		Each permutation $\pi \in T$ is uniquely determined by a partition $\alpha_1 \alpha_2 \dots \alpha_k = \alpha$, assuming $|\alpha_i| \ge 1$ for all $i \in [k]$. Each such partition can be produced by cutting $\alpha$ at any subset of the $n-1$ available positions.
	\end{proof}
	
	Second, we show that $L$ is a greedoid, which implies that topological heapsort is not optimal on greedoids (not even with an additive $\fO(n)$ term).
	\begin{proposition}
		$L$ is a greedoid.
	\end{proposition}
	\begin{proof}
		By definition, $L$ is accessible. We need to prove axiom (\ref{prop:gr:exchange}). We start with the following observation. For each word $\beta \in L$ with length $|\beta| = \ell$, we have $a_{\ell-1} \in \support{\beta}$, and $a_p \in \support{\beta}$ for some $p \ge \ell$.
		
		Now take some $\beta, \gamma \in L$ with $\ell \coloneq |\beta|$ and $|\gamma| > \ell$. We need to show that there is some $x \in \support{\gamma}$ such that $\beta x \in L$. First, suppose $\beta = \rot{\alpha_1} \rot{\alpha_2} \dots \rot{\alpha_k}$ for some prefix $\alpha_1 \alpha_2 \dots \alpha_k$ of $\alpha$ (i.e., $\beta$ ends at a ``cut position''). By our observation and since $|\gamma| \ge \ell+1$, we have $a_p \in \support{\gamma}$ for some $p \ge \ell+1$. Since $\beta a_p \in L$, the axiom holds.
		
		Now suppose that $\beta = \rot{\alpha_1} \rot{\alpha_2} \dots \rot{\alpha_k} a_p a_q a_{q+1} \dots a_{\ell-1}$ for some $p > \ell$ and $q \le \ell$ (i.e., $\beta$ does \emph{not} end at a ``cut position''). Then $\beta a_{\ell} \in L$, and $a_\ell \in \support{\gamma}$ by our observation, so we are done.
	\end{proof}
	
	The author has not been able to determine whether the ITB is tight for small greedoids (with $2^{o(n)}$ permutations); this is an interesting open question.
	
	\subsection{Non-monotone precedence systems}
	
	We now explore two generalizations of antimatroids obtained by removing one of the two conditions given in our informal definition (see \cref{sec:intro}), namely monotonicity and order-insensitivity. We start with the first. A \emph{non-monotone precedence system (NMPS)} $S = \{p_x\}_{x \in \Sigma}$ is defined just like a monotone precedence system, except that the precedence functions need not be monotone. In particular, we define the language generated by $S$ in the same way:
	\[ \Lang(S) = \{ x_1 x_2 \dots x_k \mid \forall i \in [k] : p_{x_i}(\{x_1, x_2, \dots, x_{i-1}\}) = 1 \}. \]
	
	Removing the monotonicity condition obviously increases the power of our model considerably; however, observe that not all languages are generated by an NMPS. For example, any NMPS-language on $\Sigma = \{a,b,c,d\}$ that contains the permutations $abcd$, $abdc$, and $bacd$ must also contain $badc$.
	
	We show that all greedoids are NMPS-languages, which immediately implies that topological heapsort is not optimal for NMPS-languages.
	
	\begin{proposition}
		Let $G$ be a greedoid. Then there exists an NMPS that generates $G$.
	\end{proposition}
	\begin{proof}
		Define the NMPS $S = \{p_x\}_{x \in \Sigma}$ as follows. Let $p_x(Y) = 1$ if and only if there exists a word $\alpha \in \swords{\Sigma}$ such that $\support{\alpha} = Y$ and $\alpha x \in G$.
		It is clear that $G \subseteq \Lang(S)$. We show that $\Lang(S) \subseteq G$.
		
		Suppose for the sake of contradiction that $\Lang(S) \not\subseteq G$, and take a minimal word $\beta \coloneq \alpha x \in \Lang(S) \setminus G$, with $\alpha \in G$ and $x \in \Sigma \setminus \support{\alpha}$. Then $p_x(\support{\alpha}) = 1$, so by definition, there exists an $\alpha' \in \swords{\Sigma}$ such that $\support{\alpha}' = \support{\alpha}$ and $\alpha' x \in G$. Using axiom (\ref{prop:gr:exchange}) with $\alpha$ and $\alpha'x$, we get that $\alpha x \in G$, a contradiction.
	\end{proof}
	
	However, not all NMPS-languages are greedoids. One example is the set of prefixes of $\{abcd, dcba\}$.
	
	Finally, it interesting to note that the information-theoretic lower bound is not tight for NMPS-languages. Indeed, let $\alpha = a_1 a_2 \dots a_n$ and consider the set $T$ of permutations on $\Sigma = \support{\alpha}$ that are obtained by moving some element of $\alpha$ to the start, i.e., let $T = \{ \alpha_i \mid i \in [n] \}$ with $\alpha_i = a_i a_1 a_2 \dots a_{i-1} a_{i+1} a_{i+2} \dots a_n$. Let $L$ be the set of prefixes of permutations in $T$. Then $L$ is generated by the NMPS $S = \{p_x\}_{x \in \Sigma}$ with $p_{a_i}(X) = 1$ iff $X = \emptyset$ or $X \supseteq \{a_1, a_2, \dots, a_{i-1}\}$.
	
	We now show that $n-1$ comparisons are required to sort $T$, even though the ITB is $\log |T| = \log n$. Consider an adversary that always answers comparisons consistently with the permutation $\alpha$. Every comparison $a_i < a_j$ for $i < j$ is consistent with every permutation except possibly $\alpha_j$. Thus, any sorting algorithm needs $n-1$ comparisons to eliminate all but one permutation.
	
	\subsection{Order-sensitive antimatroids}
	
	We now consider discarding the second property of antimatroids, that availability of an element may only depend on the \emph{set} of elements chosen so far, not their order. This requirement may seem somewhat arcane, but we show that it is necessary for an algorithm like topological heapsort to work.
	
	Let an \emph{order-sensitive monotone precedence system (OSMPS)} on a set $\Sigma$ be a set of functions $P = \{p_x\}_{x \in \Sigma}$ with $p_x \colon \swords{(\Sigma \setminus \{x\})} \rightarrow \{0,1\}$, such that each $p_x$ is \emph{monotone}, that is, for all $\alpha \in \swords{\Sigma}$ and $y \in \Sigma \setminus \support{\alpha}$, we have $p_x(\alpha) \le p_x(\alpha y)$. The \emph{language} of $P$ is defined as
	\[
	\Lang(P) = \{ x_1 x_2 \dots x_k \in \swords{\sigma} \mid \forall i \in [k] : p_{x_i}(x_1 x_2 \dots x_{i-1}) = 1. \}
	\]
	
	We now give an example of a set $T$ of permutations that can be described by an OSMPS, but is not an antimatroid, and moreover, topological heapsort does not optimally solve the $T$-sorting problem.
	
	Let $k \in \N$ and $n = k^2$, and let $\Sigma = \{ x_{i,j} \mid i,j \in [k] \}$. For each $i \in [k]$, define $\alpha_i = x_{i,1} x_{i,2} \dots x_{i,k}$, and let $A_i$ be the set of all permutations on $\{x_{i,1}, x_{i,2}, \dots, x_{i,k}\}$. Finally, let
	\[
	T = \{ \alpha_1 \alpha_2 \dots \alpha_{i-1} \beta_i \alpha_{i+1} \dots \alpha_k \mid i \in [k], \beta_i \in A_i \}.
	\]
	
	It is not hard to construct an OSMPS $S$ such that $T = \perms(S)$. On the other hand, $T$ is not an antimatroid: All elements are available at the start, so every antimatroid containing $T$ necessarily contains every permutation.
	
	The number of comparisons that topological heapsort takes for the $T$-sorting problem is $\Omega(n \log k)$ if the underlying total order is $\pi = \alpha_1 \alpha_2 \dots \alpha_n$. Indeed, at the start, all $k$ elements of $\support{\alpha}_1$ are inserted into the queue and then taken out in sorted order; this first phase takes $\Theta(k \log k)$ comparisons. After removing $x_{1,k}$, all elements of $\support{\alpha}_2$ are inserted, and the process repeats $k-1$ times, using $\Theta(k^2 \log k) = \Theta(n \log k)$ comparisons in total.
	
	However, observe that $|T| = k \cdot k!$, so $\log |T| \approx k \log k \approx \sqrt{n} \log n$. Thus, topological heapsort is not optimal. Also, the number of comparisons needed to sort $T$ is easily seen to be $\Theta(n)$; hence, the ITB is not tight for OSMPSs.

	\section{Conclusion}
	
	In this paper, we have shown that the basic variant of \emph{topological heapsort} nicely generalizes to the problem of sorting antimatroids, and its running time stays optimal. We also showed how the algorithm can be modified to be comparison-optimal in the antimatroid case, which in particular implies that the information-theoretic bound is tight for antimatroids. 
	Since topological heapsort is not optimal for further generalizations like \emph{greedoids}, perhaps antimatroids are the most general structure that can be sorted ``greedily''. The question whether the information-theoretic bound for greedoids holds, and whether greedoids can be sorted efficiently in some other way, is left open.
	
	For another interesting open question, consider an important class of antimatroids that was omitted from \cref{sec:app}.
	\emph{Convex shelling antimatroids}~\cite{Dietrich1987,KorteSchraderEtAl1991,BjoernerZiegler1992} are obtained by progressively removing points from the convex hull of a given point set. A candidate data structure for such an antimatroid would be a certain kind of \emph{decremental convex hull} data structure. The total running time of this data structure cannot be linear as in the other examples, since computing the convex hull can take superlinear time~\cite{KirkpatrickSeidel1986}. Is there a candidate data structure for the convex shelling antimatroid that is fast enough to sort optimally?

	\bibliography{info}

	\appendix
	
	\section{Monotone precedence systems characterize antimatroids}\label{app:mps-char}
	
	Recall the two axioms that characterize an antimatroid $A$:
	\begin{enumerate}[(i)]
		\item For all $\alpha \in \swords{\Sigma}$ and $x \in \Sigma \setminus \support{\alpha}$, if $\alpha x \in A$, then $\alpha \in A$.
		\item For each $\alpha, \beta \in A$ with $\support{\alpha} \not\subseteq \support{\beta}$, there is some $x \in \support{\alpha} \setminus \support{\beta}$ such that $\beta x \in A$.
	\end{enumerate}
	
	\begin{restatable}{proposition}{restateMPSChar}\label{p:mps-char}
		A non-empty simple language $A \subseteq \swords{\Sigma}$ is an antimatroid if and only if there exists an MPS $S$ on $\Sigma$ with $\Lang(S) = A$.
	\end{restatable}
	\begin{proof}
		We start with the ``only if'' part. Let $A$ be an antimatroid. Define the monotone precedence system $S = \{p_x\}_{x \in \Sigma}$ as $p_x(X) = 1$ iff there exists a word $\alpha x \in A$ such that $\support{\alpha} \subseteq X$. We show that $\Lang(S) = A$. By definition, we have $A \subseteq \Lang(S)$. Now take some word $\alpha = a_1 a_2 \dots a_k \in \Lang(S)$. We show $\alpha \in A$ by induction on $k$.
		
		If $k = 0$, then $\alpha = \emptyword$. Every non-empty antimatroid contains the empty word by axiom~(\ref{prop:am:acc}).
		Now suppose $k > 0$, and let $\alpha' = a_1 a_2 \dots a_{k-1}$. By definition, we have $\alpha' \in \Lang(S)$, so $\alpha' \in A$ by induction. Further observe that $\alpha \in \Lang(S)$ implies that $p_{a_k}(\support{\alpha}') = 1$. By definition of $S$, this means there must be some $\alpha''$ with $\support{\alpha}'' = \support{\alpha}'$ and $\alpha'' a_k \in A$. Finally, using axiom (\ref{prop:am:exchange}) with $\alpha'$ and $\alpha'' a_k$ implies that $\alpha = \alpha' a_k \in A$.
		
		We now show the ``if'' part. Let $S = \{p_x\}_{x \in \Sigma}$ be a monotone precedence system. We need to show that $\Lang(S)$ is an antimatroid. Axiom (\ref{prop:am:acc}) holds by definition.
		
		To prove axiom (\ref{prop:am:exchange}), let $\alpha, \beta \in \Lang(S)$ with $\support{\alpha} \not\subseteq \support{\beta}$. Let $\gamma$ be the longest prefix of $\alpha$ such that $\support{\gamma} \subseteq \support{\beta}$. Then, there exists $x \in \Sigma \setminus \support{\beta}$ and $\alpha' \in \swords{\Sigma}$ such that $\alpha = \gamma x \alpha'$. We claim that $\beta x \in \Lang(S)$. Since $\beta \in \Lang(S)$, it is enough to show that $p_x(\support{\beta}) = 1$.
		
		Since $\gamma x \alpha' \in \Lang(S)$, we also have $\gamma x \in \Lang(S)$, which implies $p_x(\support{\gamma}) = 1$. Since $\support{\gamma} \subseteq \support{\beta}$, monotonicity of $p_x$ implies $p_x(\support{\beta}) = 1$, as desired.
	\end{proof}

	We have not been able to find the characterization of antimatroids via monotone precedence systems in the literature; however, a slightly different formalism called \emph{alternative precedence systems} is known to describe antimatroids~\cite{BjoernerZiegler1992} (cf.\ \cref{sec:app:formulae}). Equivalence between alternative precedence systems and monotone precedence systems is quite easy to see.
	
	\section{Vertex search and distance orderings}\label{app:dijkstra}
	
	Given an edge-weighted directed graph $(G,w)$ with a designated \emph{source} $s$, a \emph{distance ordering} of $(G,s,w)$ is an ordering of the vertices according to distance from $s$, with ties broken arbitrarily. We assume all edge weights are positive and every vertex is reachable from $s$. Let $D_{G,s}$ denote the set of possible distance orderings produced by a weight function $w$ on $G$.
	We now show that $D_{G,s}$ is exactly the set of permutations of the vertex search antimatroid $\AMVS_{G,s}$ of $G$ with root $s$ (see \cref{sec:dist-order}).
	
	\begin{proposition}\label{p:do-vs}
		Let $G$ be a directed graph and let $s \in V(G)$ such that each vertex of $G$ is reachable from $s$. Then $D_{G,s} = \perms(\AMVS_{G,s})$.
	\end{proposition}
	\begin{proof}
		In the following, let $d(u,v)$ denote the distance between two vertices $u$ and $v$. Recall that $\AMVS_{G,s}$ is characterized by the monotone precedence system $S = \{p_x\}_{x \in \Sigma}$ with $p_x(X) = 1$ if $x = s$, or otherwise if $X$ contains at least one in-neighbor of $x$ in $G$.
		
		We first show that $D_{G,s} \subseteq \perms(\AMVS_{G,s})$. Let $\pi = v_1 v_2 \dots v_n$ with $v_1 = s$ be a distance ordering of $G$ with some weight function $w$. Let $i \ge 2$.
		We need to show that $p_{v_i}(V_{i-1}) = 1$, where $V_{i-1} = \{v_1, v_2, \dots, v_{i-1}\}$.
		
		Let $u_1 u_2 \dots u_k$ be a shortest path from $s = u_1$ to $v_i = u_k$. Then $d(s,u_{k-1}) < d(s,v_i)$, since all weights are positive. This implies $u_{k-1} \in V_{i-1}$. Since there is an edge from $u_{k-1}$ to $u_k = v_i$, we have $p_{v_i}(V_{i-1}) = 1$, as required.
		
		We now show that $D_{G,s} \supseteq \perms(\AMVS_{G,s})$.
		Let $\pi = v_1 v_2 \dots v_n \in \perms(\AMVS_{G,s})$. We first show that for each $v$, there exists a path from $s$ to $v$ that is monotonously increasing w.r.t.\ $\prec_\pi$. Indeed, for each $v_i \neq s$, we have $p_{v_i}(\{v_1, v_2, \dots, v_{i-1}\} = 1$, which implies that one of $v_1, v_2, \dots, v_{i-1}$ is an in-neighbor of $v_i$. By induction, the claimed monotone path exists.
		
		We now define an edge-weight function $w$ as $w(v_i,v_j) = j-i$ if $j > i$, and $w(v_i,v_j) = n$ if $i < j$. From the existence of $\pi$-monotone paths, we have $d(s, v_i) \le i-1$ for each $i \in [n]$. Moreover, each non-monotone path has length at least $n$. This implies that $d(s, v_i) = i-1$ for each $i \in [n]$ so $L$ is indeed a distance ordering.
	\end{proof}
	
	\subsection{Universal optimality of Dijkstra's algorithm}\label{sec:dijkstra}
	
	\newcommand{\UODijPaper}{Haeupler, Hladík, Rozhoň, Tarjan, and Tětek~\cite{HaeuplerHlad'ikEtAl2024}}
	\newcommand{\UODijSimpPaper}{van der Hoog, Rotenberg, and Rutschman~\cite{HoogRotenbergEtAl2025}}
	
	\Cref{p:do-vs} can be used to give a proof of universal optimality of Dijkstra's algorithm. We give the details below. We will again avoid formally defining the working-set property for the proof itself, though we discuss two of its variants in \cref{sec:wsprops}.
	
	The \emph{distance ordering problem} (DOP) consists of finding a distance ordering of an input graph $(G,s,w)$. Fixing $G$ and $s$, the information-theoretic lower bound tells us that any algorithm solving DOP on $(G,s)$ needs to perform at least $\log |D_{G,s}|$ comparisons in the worst-case (i.e., with a worst-case weight function $w$).\footnote{We use here the simple observation that every $\pi \in D_{G,s}$ is the \emph{unique} distance ordering of $(G,s,w)$ for some $w$.}
	An algorithm that solves DOP with running time $\fO( |E(G)| + \log |D_{G,s}| )$ for each input $(G,s,w)$ is called \emph{universally optimal}. \UODijPaper{} showed that Dijkstra's algorithm equipped with a working-set heap is univerally optimal, and \UODijSimpPaper{}, independently from this paper, showed that a weaker working-set property is sufficient.\footnote{Both papers also showed comparison optimality, which we omit here.} Our proof works with the latter, simple working-set property (see \cref{sec:wsprops}).
	
	For Dijkstra's algorithm, our priority queue $Q$ contains vertex-distance pairs $(v,d)$. We assume that the following operations are available: $Q.\Call{extract-min}{}$ returns $(v,d)$ with minimum $d$; $Q.\Call{insert}{v,d}$ inserts $(v,d)$ into the queue, assuming $v$ is not in the queue; $Q.\Call{decrease-key}{v,d}$ sets the distance associated to $v$ to $d$, assuming $(v,d')$ is contained in $Q$ with $d' \ge d$; and $Q.\Call{get-key}{v}$ returns $d$ if $(v,d)$ is in $Q$, or returns $\bot$ if $v$ is not in $Q$. The last operation is not usually defined for priority queues, but can easily be implemented in $\fO(1)$ time with a table lookup.
	
	\begin{algorithm}
		\caption{Dijkstra's algorithm, adapted for the distance-ordering problem.}
		\begin{algorithmic}
			\Statex \textbf{Input:} Digraph $G$, starting vertex $s$, edge weight function $w \colon E(G) \rightarrow \R_+$
			\State Initialize working-set priority queue $Q$ with $\{(s,0)\}$
			\While{$Q$ is not empty}
				\State $(u,d) \gets Q.\Call{extract-min}{}()$
				\State Report $u$
				\ForEach{$v \in V(G)$, $(u,v) \in E(G)$}
					\If{$Q.\Call{get-key}{v} = \bot$}
						\State $Q.\Call{insert}{v,d+w(u,v)}$
					\ElsIf{$Q.\Call{get-key}{v} > d + w(u,v)$}
						\State $Q.\Call{decrease-key}{v, d+w(u,v)$}
					\EndIf
				\EndFor
			\EndWhile
		\end{algorithmic}
	\end{algorithm}
	
	The \emph{transcript} of a run of Dijkstra's algorithm is defined in the same way as for topological heapsort, i.e., as the sequence $Q_0, Q_1, \dots, Q_n$, where $Q_0$ is the set of elements initially in $Q$, and $Q_i$ is the set of elements in $Q$ after the $i$-th step. The following lemma is critical for our proof.
	
	\begin{lemma}\label{p:dijkstra-transcript}
		Consider an input $(G,s,w)$ with a unique distance ordering $\pi$. Then, running Dijkstra's algorithm on $(G,s,w)$ and running topological heapsort on the vertex search antimatroid $\AMVS_{G,s}$ with a comparison oracle for $\pi$ produces the same transcript.
	\end{lemma}
	\begin{proof}
		Let $\pi = x_1 x_2 \dots x_n$. Let $(Q_0, Q_1, \dots, Q_n)$ be the transcript for topological heapsort on $\AMVS_{G,s}$, $\pi$. Then, by definition, $Q_0 = \{s\} = \{x_1\}$, and $Q_i$ contains precisely those vertices with an in-neighbor in $\{x_1, x_2, \dots, x_i\}$. Considering the fact that Dijkstra's algorithm outputs $x_1, x_2, \dots, x_n$, it is easy to see that its transcript is also $(Q_0, Q_1, \dots, Q_n)$.
	\end{proof}
	
	Say a working-set priority queue executes a certain sequence $\sigma$ of operations \Call{extract-min}{} and \Call{insert}{} in total time $t$. We then define a \emph{decrease-key working-set priority queue} as a priority queue that can execute the same sequence $\sigma$, with $m$ \Call{decrease-key}{} operations inserted anywhere, in time $\fO(t+m)$. In words, \Call{decrease-key}{} takes amortized constant time and does not affect the running time the other operations.
	
	Now take some input $(G,s,w)$ with a unique distance ordering $\pi$. Topological heapsort on $\AMVS_{G,s}$, $\pi$ runs in time $\fO(\log |\perms(\AMVS_{G,s})|)$. Since Dijkstra's algorithm performs no more than $|E(G)|$ \Call{decrease-key}{} operations, its running time is $\fO( |E(G)| + \log |\perms(\AMVS_{G,s})|)$ by \cref{p:dijkstra-transcript}. Finally, \cref{p:do-vs} implies that $\perms(\AMVS_{G,s}) = D_{G,s}$, so Dijkstra's algorithm with a decrease-key working-set priority queue is indeed universally optimal.
	
	\subsection{Working-set properties}\label{sec:wsprops}
	
	We now discuss two relevant variants of the working-set property. For both, the following definition is useful. Fix a sequence of operations on a priority queue $Q$ (from now on, we always assume each element is inserted at most once). We say time $t$ is the time directly after the $t$-th operation.
	If $x$ is contained in $Q$ at time $t$, then the \emph{working set of $x$ at time $t$}, written $W(x,t)$ contains each element $y$ such that (1) $y$ is contained in $Q$ at time $t$, and (2) $y$ was inserted into $Q$ no later than $x$ (in particular, we have $x \in W(x,t)$).
	
	The first variant, the \emph{intermediate working set property}, has been used by Haeupler et al.\ in the first paper on universal optimality of Dijkstra's algorithm~\cite{HaeuplerHlad'ikEtAl2024}, as well as in the initial version of the original topological heapsort paper~\cite{HaeuplerHladikEtAl2024} (both times simply under the name \emph{working set property}). The \emph{working set size} $w(x)$ of an element $x$ is the maximum size of $W(t,x)$ for any time $t$ where $x$ is contained in $Q$. We say $Q$ has the intermediate working set property if $\Call{insert}{}$ and $\Call{decrease-key}{}$ have amortized running time $\fO(1)$, and $\Call{extract-min}{}$ has running time $\fO(1 + \log w(x) )$, where $x$ is the returned element.
	
	Haeupler et al.\ developed a priority queue with the intermediate working set property~\cite{HaeuplerHlad'ikEtAl2024}. Alternatively, Iacono showed that the classical \emph{pairing heaps}~\cite{FredmanEtAl1986} have the strong working set property if \Call{decrease-key}{} is not used (as for topological heapsort).
	
	The second property, called here the \emph{weak working set property}, is defined as follows: Let $w'(x)$ be the size of the \emph{union} of all $W(t,x)$ for any time $t$ where $x$ is in $Q$. In other words, $w'(x)$ is the number of elements inserted into the priority queue while $x$ is present. In the final version of the topological heapsort paper~\cite{HaeuplerHladikEtAl2025a} the proof was changed to only require the weak working set property. By extension, topological heapsort for antimatroids also works with the weak working set property, and the proof above shows that the weak working set property is sufficient for universal optimality of Dijkstra's algorithm as well.
	
	Independently of this work, \UODijSimpPaper{} also showed that Dijkstra's algorithm is universally optimal with a weak working set priority queues (called \emph{timestamp optimality} by them). Note that they additionally showed how to achieve \emph{comparison optimality} with weak working set priority queues, and gave a greatly simplified priority queue implementation that has the weak working set property and supports \Call{decrease-key}{}.

	\section{Decremental simplicial vertices}\label{sec:dec-simplicial}

	We show how to dynamically maintain the set of simplicial vertices in a chordal graph, under simplicial vertex deletions.
	
	We first need a few definitions. Let $G$ be a chordal graph with $n$ vertices and $m$ edges. A clique in $G$ is \emph{maximal} if it is not strictly contained in any other clique.
	Observe that a vertex is simplicial if and only if it is contained in only one maximal clique.
	
	A \emph{clique tree} of $G$ is a tree $T$ such there is a bijection $K$ between $V(T)$ and the set of maximal cliques of $G$, and the following holds: For all $x_1, x_2, x_3 \in V(T)$ such that $x_2$ lies on the path between $x_1$ and $x_2$, we have $K(x_1) \cap K(x_3) \subseteq K(x_2)$.\footnote{Observe that a clique tree is a special type of \emph{tree decomposition} as is used for treewidth~\cite{RobertsonSeymour1986}.} It is known that a clique tree on $G$ can be constructed in $\fO( m + n )$ time, and that $\sum_{x \in V(T)} |K(x)| \in \fO( m + n )$~\cite{BlairPeyton1993,ChandranEtAl2003}.
	
	Our data structure is a clique tree $T$. Each clique maintains a linked list of its vertices, and each vertex $v \in V(G)$ maintains a pointer to some vertex $x \in V(T)$ such that $v \in K(x)$. Additionally, each edge $e = \{x,y\}$ stores the number $s_e = |K(x) \cap K(y)|$, and each vertex $x \in V(T)$ stores the list of incident edges, sorted by $s_e$.  Observe that a given clique tree can be preprocessed as such in $\fO(m+n)$ time.
	
	At the start, we identify each simplicial vertex by checking if it is contained in more than one maximal clique. Now suppose a simplicial vertex $v$ is removed. Then $v$ is contained in only one maximal clique $K(x)$. Remove $v$ from $K(x)$. We now need to check if $K(x)$ is still maximal. If it is not, then there is some neighbor $y$ of $x$ in $T$ such that $K(x) \subseteq K(y)$. This is the case if and only if $s_{\{x,y\}} = |K(x)|$ for some edge $\{x,y\} \in E(T)$. Checking this condition can be done in constant time, since we store all values $s_e$ of incident edges in sorted order.
	
	If $K(x) \subseteq K(y)$, then we contract the edge $\{x,y\}$ in $T$. All necessary updates can be done in $|K(x)|$ time. Observe that the amount of work we do for each vertex removal is proportional to the amount of data we remove from the data structure, hence the total running time is $\fO(m+n)$.
\end{document}